\documentclass[twoside,11pt]{article}
\usepackage{amsmath}
\usepackage{amssymb}
\usepackage{amsfonts}
\usepackage{cite}
\usepackage{mathrsfs}
\usepackage{graphicx}
\usepackage{url}
\usepackage{color}
\usepackage{hyperref}
\hypersetup{colorlinks=true,linkcolor=[rgb]{0,0,0.75},citecolor=[rgb]{0,0,0.75},urlcolor=[rgb]{0,0,0.75}}
\usepackage{amsthm}
\usepackage{thmtools}
\usepackage{fullpage}
\usepackage[parfill]{parskip}

\allowdisplaybreaks

\declaretheorem[name=Theorem]{theorem}
\declaretheorem[name=Lemma,numberwithin=section]{lemma}

\newcommand{\ceil}[1]{\left\lceil #1\right\rceil}

\newcommand{\lift}[1]{\widehat{#1}}
\newcommand{\RE}{\mathbb{R}}

\newcommand{\eps}{\varepsilon}
\newcommand{\ST}{\,:\,}
\newcommand{\SP}{\kern+1pt}
\newcommand{\Transpose}{\intercal}
\newcommand{\bd}{\partial}
\newcommand{\etal}{\textit{et al.}}
\DeclareMathOperator{\diam}{diam}

\DeclareMathOperator{\dist}{dist}

\DeclareMathOperator{\ray}{ray}

\DeclareMathOperator{\interior}{int}

\newcommand{\supp}{\mathrm{supp}}
\newcommand{\rep}{\mathrm{rep}}
\newcommand{\depth}[1]{\Delta_{#1}}

\newcommand{\Gradient}{\nabla}
\newcommand{\Hess}{\Gradient^2}
\newcommand{\dm}{d}
\newcommand{\adm}{\widetilde{d}}
\newcommand{\Q}{\mathcal{Q}}
\newcommand{\DQ}{\widetilde{d}_\Q}
\newcommand{\bdOmega}{\partial \kern+1pt \Omega}

\title{Smooth Distance Approximation\thanks{A condensed version of this paper appeared in the 31st Annu.\ European Sympos.\ Algorithms, 2023.}} 

\author{Ahmed Abdelkader\thanks{Conducted in part while the first author was a postdoctoral fellow at the University of Texas at Austin and completed before he joined Google LLC.} \\
	Google LLC, Mountain View, California, USA \\
	\href{mailto:ahmadabdolkader@gmail.com}{ahmadabdolkader@gmail.com} \\
	\and
	David M. Mount\\
	Dept.\ of Computer Science and \\
	Inst.\ for Advanced Computer Studies \\
	University of Maryland, College Park, Maryland USA \\
	\href{mailto:mount@umd.edu}{mount@umd.edu}
}

\date{}

\begin{document}

\maketitle

\begin{abstract}
Traditional problems in computational geometry involve aspects that are both discrete and continuous. One such example is nearest-neighbor searching, where the input is discrete, but the result depends on distances, which vary continuously. In many real-world applications of geometric data structures, it is assumed that query results are continuous, free of jump discontinuities. This is at odds with many modern data structures in computational geometry, which employ approximations to achieve efficiency, but these approximations often suffer from discontinuities. 

In this paper, we present a general method for transforming an approximate but discontinuous data structure into one that produces a smooth approximation, while matching the asymptotic space efficiencies of the original. We achieve this by adapting an approach called the partition-of-unity method, which smoothly blends multiple local approximations into a single smooth global approximation.

We illustrate the use of this technique in a specific application of approximating the distance to the boundary of a convex polytope in $\RE^d$ from any point in its interior. We begin by developing a novel data structure that efficiently computes an absolute $\eps$-approximation to this query in time $O(\log (1/\eps))$ using $O(1/\eps^{d/2})$ storage space. Then, we proceed to apply the proposed partition-of-unity blending to guarantee the smoothness of the approximate distance field, establishing optimal asymptotic bounds on the norms of its gradient and Hessian.
\end{abstract}


\section{Introduction} \label{sec:intro}

The field of computational geometry has largely focused on computational problems with discrete inputs and outputs. Discrete structures are often used to represent geometric objects that are naturally continuous. Examples include using triangulated meshes to represent smooth surfaces, Voronoi diagrams to represent distance maps, and various spatial partitions for answering ray-shooting queries. Due to the high computational complexities involved, researchers often turn to approximation algorithms. Unfortunately, in retrieval problems, efficient approximation is often achieved at the expense of continuity. 

To make this more precise, consider the common example of distance functions. For a given set $S \subseteq \RE^d$ (which may be discrete or continuous), a natural \emph{distance map} over $\RE^d$ arises as:
\[
    \dm_{S}: x \mapsto \inf_{p \in S} \|x - p\|,
\]
where $\|\cdot\|$ denotes the Euclidean norm. In turn, the distance map gives rise to the following query problem. Given a query point $x \in \RE^d$, the objective is to compute $\dm_{S}(x)$ efficiently from a data structure of low storage.

It is well known that answering the distance query can be reduced to computing the Voronoi diagram of $S$. Unfortunately, beyond special low-dimensional cases, the combinatorial complexity of the Voronoi diagram grows too fast for practical use. For this reason, much work has focused on data structures for approximate nearest neighbor (ANN) searching~\cite{AMNSW98, AFM17a, GIM99, Har01, HIM12}. Given any $\eps > 0$, an $\eps$-ANN data structure returns a point that is within a factor of $1+\eps$ of the true closest distance.

While approximate nearest-neighbor searching is clearly related to approximating the distance map, there are fundamental differences between the two problems. The distance map induced by any set is clearly continuous (and indeed it is 1-Lipschitz continuous~\cite{BaS99}). As two query points converge on a common location, their respective distances to $S$ must also converge. The same cannot be said for any of the existing approaches based on approximate nearest neighbor searching. The ANN distances reported for two query points can differ by an amount that is arbitrarily larger than the distance between the two query points. In Section~\ref{sec:witness}, we will show that this is not merely an artifact of the design of these data structures; it is unavoidable.

Answering distance queries efficiently is key to many applications including motion planning~\cite{yershova2007improving}, surface reconstruction~\cite{Amenta1999,HKM16}, physical modeling~\cite{Osher2003}, and data analysis~\cite{Aurenhammer:2013:VDD, CompTopoBook}. Discontinuities can result in various sorts of aberrant behaviors. This is because queries are generated adaptively in a feedback loop, where answers to earlier queries are used to determine subsequent queries. Consider, for example, a navigation system that is trying to precisely dock two crafts moving in space. Discontinuities in the distance map can alter the behavior of the feedback process, resulting in jittering, oscillations, and even infinite looping (see examples in Section~\ref{sec:witness}). 

This motivates the main question considered in this paper: \emph{Does there exist a data structure that answers distance queries approximately so that the induced distance function is continuous?} Ideally, the distance function should also be smooth, characterized by bounds on the norm of its gradient and Hessian. Note that this is quite different from approximate nearest-neighbor searching, where the objective is to find a point that approximates the closest distance. Here, the objective is approximate the distance itself.

Applications of distance queries include collision detection~\cite{Brochu:2012:EGE}, penetration depth~\cite{ZHANG20143}, robot navigation~\cite{tiwari2022touch, lopez2013gap}, shape matching~\cite{al2013continuous}, and density estimation~\cite{pmlr-v206-marchetti23a}. Often, the set $S$ arises as a discrete point set obtained by sampling an underlying surface. Implicit representations of surfaces~\cite{bloomenthal1997introduction}, based on approximating the induced distance map, have recently witnessed significant developments based on deep neural networks~\cite{park2019deepsdf, gropp2020implicit, chibane2020neural}, where the properties of learned distance fields are yet to be fully understood~\cite{pmlr-v139-lipman21a,spelunking_TOG22}.

In this paper we present a general approach for smooth approximation from traditional non-continuous data structures. This is achieved through a process called \emph{blending}, where discrete local approximations are combined to form a smooth function. Our method is loosely based on the \emph{partition-of-unity method} (see, e.g., Melenk and Babu{\v s}ka \cite{Melenk1996}). The approach involves constructing an open cover of the domain by overlapping patches, computing a local approximation within each patch, and then blending these approximations together by associating a smooth weighting function with each patch (see Section~\ref{sec:pou} for details).

Unfortunately, a direct adaptation of these methods does not yield an efficient solution. To the best of our knowledge, existing work on partition-of-unity methods for distance approximation have not considered the asymptotic efficiency of the resulting access structures. These works have typically involved blending over relatively simple spatial decompositions, such as grids \cite{Stein1970} and balanced quadtrees \cite{Ohtake:2003}. The covering elements employed in the blending were naturally fat, that is, \emph{isotropic}. These subdivisions are particularly suitable for blending, but they lack the flexibility needed to achieve the highest levels of efficiency. Moreover, we are not aware of prior results on the asymptotic interplay between approximation and smoothness. (We refer the interested reader to recent works in the finite element literature on anisotropic~\cite{ZANDER2022103700} and high-dimensional~\cite{KOPP2022115575} refinements.) In this paper, we adopt the partition-of-unity approach to perform smooth blending for distance maps while achieving asymptotic complexity bounds that match the best existing approximation algorithms. Our results will be presented in Section~\ref{sec:medial-dist}.

\subsection{On Discontinuities and Witnesses} \label{sec:witness}

To better understand how discontinuities arise, it is useful to understand the general structure of most data structures for answering distance queries. Space is subdivided into regions, or \emph{cells}. This is either done explicitly by defining the subdivision over the query range or implicitly by viewing the data structure abstractly as a decision tree and associating each leaf of the tree with the subset of query points that land in this leaf due to the search process. Queries are answered by determining the cell (or cells) that are relevant to the answer, and accessing distance information for each cell. When the query point moves from one cell to another, even infinitesimally, different distance information is accessed, and the computed distance may change discontinuously.

For example, consider four point sites $P = \{p_1, p_2, p_3, p_4\}$ in $\RE^2$. Suppose that we construct an $\eps$-ANN data structure based on a subdivision into rectangular cells (see Figure~\ref{fig:witness}(a)). We assume that each cell stores a single site of $P$, called a \emph{representative}, that serves as an $\eps$-ANN for every query point lying in this cell, and assume further that the representatives have been chosen as shown in the figure, with $q_i$'s representative being $p_i$. Suppose that a gradient descent algorithm is run using this structure. Starting from an initial position (e.g., $q_i$), the descent takes a step towards the cell's representative ($p_i$). If the representatives and step sizes are chosen as in the figure, the descent could loop infinitely.

\begin{figure}[htbp]
   \centering\includegraphics[scale=0.40,page=1]{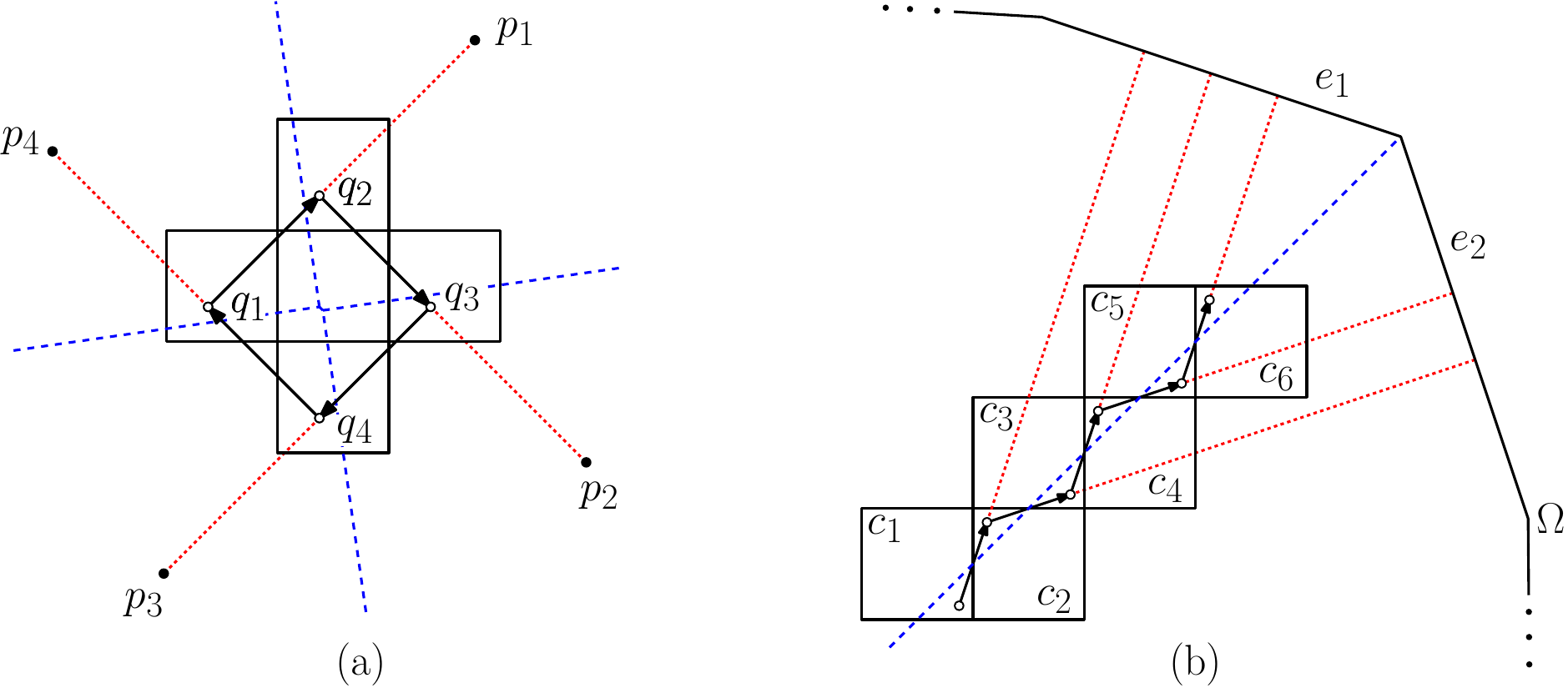}
   \caption{Problems with witness-based distance approximation: (a) infinite loops and (b) jittering. (The dashed blue lines bound the Voronoi cells of the sites, and the dotted red lines indicate the direction to the closest site.)}
   \label{fig:witness}
\end{figure}

In Figure~\ref{fig:witness}(b), we consider another distance function computed with respect to the boundary of a convex object $\Omega$. Cells $c_1$, $c_3$, and $c_5$ are assigned edge $e_1$ as representative, and cells $c_2$, $c_4$, and $c_6$ are assigned $e_2$. If at each point we walk towards the closest edge to the cell's centroid, the path oscillates or ``jitters'' between the two contenders.

In both of these examples, we assume a standard model in which each cell stores a \emph{witness} to an approximate nearest neighbor, and the distance function returns the distance from the query point to this witness. Let $\Q$ be a function that maps query points to witnesses (presumably based on the cell containing the query point), and let $\DQ$ denote the induced distance function $\DQ(x) = \|x - \Q(x)\|$. Such an approach is said to be \emph{witness-based}. The following lemma shows that any witness-based method that fails to be exact cannot be both continuous and accurate with respect to relative errors.

\begin{lemma} \label{lem:impossibility}
If a witness-based distance function $\DQ$ for a finite point set $P \subset \RE^d$ is inexact at even one point, it cannot be both continuous and provide a finite bound on relative errors.
\end{lemma}

\begin{proof}
Suppose towards a contradiction that $\DQ$ is continuous, guarantees a relative error of at most $c$ for some $c > 0$, but there exists a point $x \in \RE^d$ such that $\DQ(x) > \dm_P(x)$. In particular, we may select an arbitrarily small $\delta > 0$ such that $\DQ(x) > \dm_P(x) + \delta$. Let $p \in P$ denote a nearest neighbor of $x$ and consider how the value $\DQ$ varies as we walk from $x$ to $p$ along the line segment $\overline{x p}$. More precisely, letting $u$ be a unit vector directed from $x$ to $p$, define $x(t) = x + t \cdot u$, and $\DQ(t) = \DQ(x(t))$. Except at a finite number of transition points where the witness changes, the derivative of $\DQ(t)$ with respect to $t$ cannot be smaller than $-1$. (A derivative of $-1$ occurs when we are walking straight towards the current witness, and otherwise it is strictly larger.) Since the function is continuous, its value does not change at transition points. It follows that as we travel a distance of $t \leq \dm_P(x)$ from $x$ to $p$, the value returned by $\DQ$ cannot decrease by an amount more than $t$. Setting $t = \dm_P(x) - \delta/c$, we conclude that
\[
    \DQ(x(t))
        ~ \geq ~ \DQ(x) - t
        ~ >    ~ (\dm_P(x) + \delta) - \left( \dm_P(x) - \frac{\delta}{c} \right)
        ~ =    ~ \frac{(1+c) \delta}{c}.
\]
But, $\dm_P(x(t)) = \dm_P(x) - t = \delta/c$, implying that the relative error is 
\[
    \frac{\DQ(x(t)) - \dm_P(x(t))}{\dm_P(x(t))}
        ~ = ~ \frac{\DQ(x(t))}{\dm_P(x(t))} - 1
        ~ > ~ \frac{(1+c) \delta}{c} \cdot \frac{c}{\delta} - 1
        ~ = ~ (c + 1) - 1
        ~ = ~ c,
\]
a contradiction.
\end{proof}

\subsection{Main Result} \label{sec:medial-dist}

For the sake of concreteness, we will illustrate our approach to producing smooth approximate distance functions in a specific application which is fairly simple, but still new. Let $\Omega$ denote a convex polytope in $\RE^d$, and let $\diam(\Omega)$ denote its diameter and $\bdOmega$ its boundary. We further assume that $\Omega$ is represented as the intersection of $n$ halfspaces. Given a point $x \in \Omega$, we define the \emph{boundary distance function} $\dm_{\bdOmega}(x)$ as the Euclidean distance to $x$'s closest point on $\bdOmega$. To simplify notation, we will refer to this as $\dm_{\Omega}(x)$ (see Figure~\ref{fig:media-axis}(a)). Our objective is to efficiently evaluate an $\eps$-approximation $\adm_{\Omega}$ for any given query $x \in \Omega$, while guaranteeing smoothness (i.e., continuity and norm bounds on the gradient and Hessian).

\begin{figure}[htbp]
   \centering\includegraphics[scale=0.40]{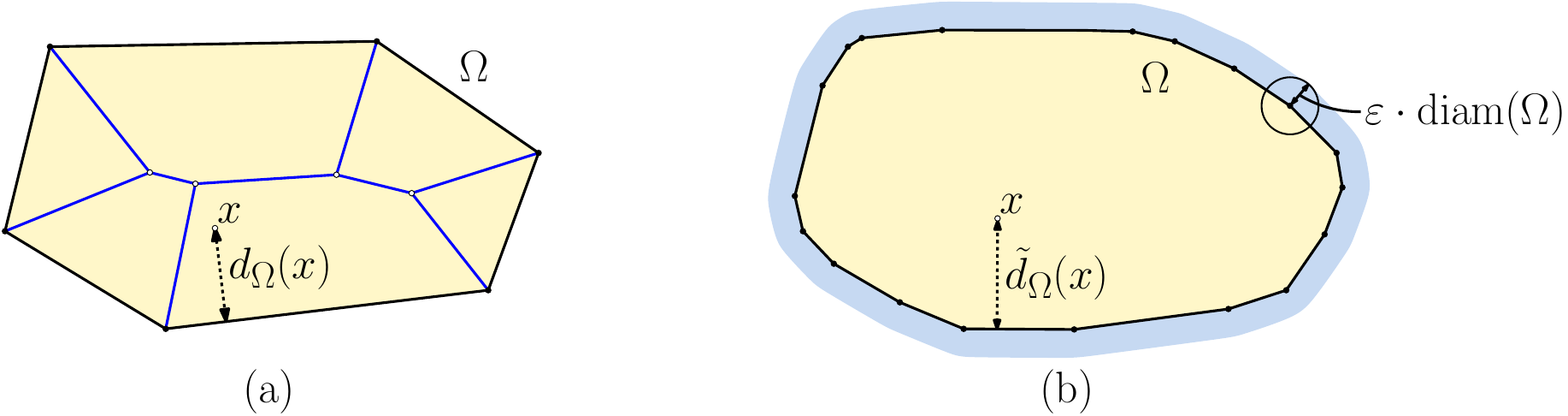}
   \caption{(a) The medial axis of $\Omega$ and the boundary distance function $d_{\Omega}$ and (b) approximating the boundary distance in terms of absolute errors with parameter $\eps > 0$.}
   \label{fig:media-axis}
\end{figure}

By convexity, if $x$ lies in $\Omega$'s interior, $\interior(\Omega)$, its closest point on the boundary lies on one of $\Omega$'s facets, that is, its faces of dimension $d-1$. Thus, in the exact setting, the distance map is determined by the Voronoi diagram of $\Omega$'s facets. The skeleton of this Voronoi diagram is known as the \emph{medial axis} or \emph{medial diagram} of $\Omega$\cite{EE98,CKM99,TDS16}. While the combinatorial complexity of the medial axis is $O(n)$ in $\RE^2$, it grows much faster in higher dimensions. It is not hard to show that medial axis corresponds to the lower-envelope of $n$ hyperplanes in $\RE^{d+1}$, with a combinatorial complexity of $\Theta(n^{\ceil{d/2}})$ in the worst case \cite{McM70}.

The obvious discrete analog to our problem is approximate polytope membership, where the data structure merely indicates whether the query point lies inside or outside the polytope, up to a Hausdorff error of $\eps \cdot \diam(\Omega)$ (see Figure~\ref{fig:media-axis}(b)). In recent work, it was shown that this problem can be solved in query time $O(\log (1/\eps))$ from a data structure using $O(1/\eps^{(d-1)/2})$ of space~\cite{AFM17b, AbM18}.

In this paper, we show how to apply the partition-of-unity method to evaluate an absolute $\eps$-approximate boundary distance function $\widetilde{d}_{\Omega}$ for a convex polytope $\Omega$ in a manner that guarantees smoothness while nearly matching the query times achieved in approximate membership queries. Specifically, we require that $|\widetilde{d}_{\Omega}(x) - \dm_{\Omega}(x)| \leq \eps \cdot \diam(\Omega)$, for all $x \in \interior(\Omega)$. Throughout we treat $\eps$ as an asymptotic quantity, and assume the dimension $d$ is a constant. Our main result is:

\begin{theorem}\label{thm:main}
Given a convex polytope $\Omega$ and an approximation parameter $\eps > 0$, there exists a smooth function $\adm_{\Omega}$ satisfying $\dm_{\Omega}(x) \leq \adm_{\Omega}(x) \leq  \dm_{\Omega}(x) + \eps \cdot \diam(\Omega)$ for all $x \in \Omega$, which can be evaluated along with its gradient from a data structure with
\[
    \text{Query time} ~ = ~ O(\log (1/\eps)) \qquad\text{and} \qquad 
    \text{Storage} ~ = ~ O(1/\eps^{d/2}).
\]
Further, the norms of the gradient and Hessian of $\adm_{\Omega}$ satisfy
\[
    \big\|\Gradient \adm_{\Omega}(x)\big\| 
        ~ = ~ O(1)
        \qquad \text{and} \qquad
    \big\|\Hess \adm_{\Omega}(x)\big\| 
        ~ = ~ O\bigg(\frac{1}{\eps}\bigg).
\]
\end{theorem}

Observe that this is almost as good as the best query and space times for approximate polytope membership~\cite{AFM17b, AbM18}, suffering just an additional factor $1/\sqrt{\eps}$ in the space bound. Our data structure can be viewed as incorporating blending into the data structure of~\cite{AbM18}. While we assume that the query point lies within $\Omega$, if this is not the case and $x$ is at distance at least $\eps \cdot \diam(\Omega)$ outside, the data structure will report this. If $x$ is external to $\Omega$ but is closer than this to the boundary, it may erroneously report an answer to the query. Our focus is on the existence of the data structure, but through the use of known constructions, it can be built in time $O(n/\eps^{O(d)})$, where $n$ denotes the number of facets of the polytope.

Let us remark on the bounds on the norms of the gradient and Hessian. Clearly, in any Euclidean distance field the directional derivative of the distance field is as high as 1 (when moving directly towards or away from the nearest point) and is never greater, that is, $\|\Gradient \dm_{\Omega}(x)\| \leq 1$. Therefore, it is reasonable that the norm of our approximate function, $\|\Gradient \adm_{\Omega}(x)\|$, is $O(1)$. The following lemma shows that the $O(1/\eps)$ upper bound on the norm of the Hessian is a necessity, up to constant factors. It establishes a lower bound in the context of a relative errors for approximating the distance to a discrete point set, but the result can be adapted to our context as well.

\begin{restatable}{lemma}{HessianLB}\label{lem:hessian-lower-bound}
Fix a set of points $P \subset \RE^d$, and let $Q$ be a smooth $\eps$-approximate distance query structure over $P$ with the associated distance $\DQ$, for any $\eps > 0$ bounding the relative error. If $|P| > 1$, then there exists a point $x \in \RE^d$ such that $\big\|\Hess \DQ(x)\big\| \geq 1/\eps$.
\end{restatable}

\begin{proof}
\newcommand{\GradientF}{\Gradient \kern-1pt f}
Given a function $f: \RE^d \to \RE$ and $\gamma \geq 0$, the assertion $\|\GradientF(a) - \GradientF(b)\| \leq \gamma$ is equivalent to saying that $\GradientF$ is $\gamma$-Lipschitz, that is, $\|\GradientF(a) - \GradientF(b)\| \leq \gamma \cdot \|a - b\|$, for all $a, b \in \RE^d$. Letting $f := \Gradient \DQ$, we will show that $f$ is $\gamma$-Lipschitz with $\gamma \geq 1/\eps$. 

\begin{figure}[htbp]
  \centerline{\includegraphics[scale=0.40]{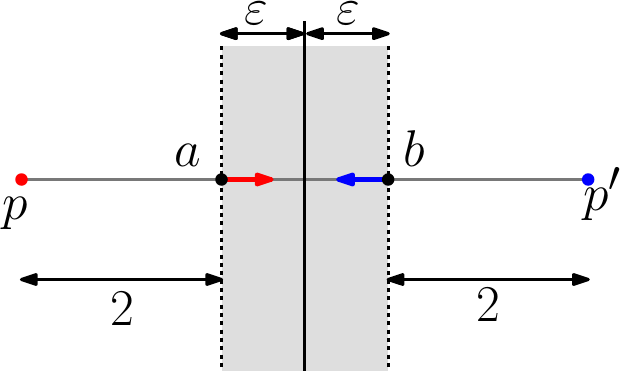}}
  \caption{Proof of Lemma~\ref{lem:hessian-lower-bound}.} \label{fig:hessian-lower-bound}
\end{figure}

Consider two sites $p$ and $p'$ such that $\|p p'\| = 2(2+\eps)$ (see Figure~\ref{fig:hessian-lower-bound}). Select points $a$ and $b$ along the segment $p p'$ on opposite sides and at distance $\eps$ from the perpendicular bisector. Observe that a query point placed at any point on the open segment $p a$ must return $p$ as the answer, since otherwise the relative error would exceed $((2+2\eps)-2)/2 = \eps$. This holds symmetrically for $p' b$. It follows that $\GradientF(a)$ and $\GradientF(b)$ are unit vectors pointing to the right and left, respectively. Hence, $\|\GradientF(a) - \GradientF(b)\|/\|a - b\| = 2/2\eps = 1/\eps$, as desired.
\end{proof}

The remainder of the paper is organized as follows. In the next section we present an overview of the partition-of-unity approach. In Section~\ref{sec:medial-apx} we present an efficient data structure for answering approximate distance queries for a convex polytope $\Omega$, but without continuity. Finally, in Section~\ref{sec:together}, we combine these to obtain the desired smooth approximation.

\section{Blending and Partition of Unity} \label{sec:pou}

The partition of unity is a standard mathematical tool for integrating local constructions into global ones~\cite{Lee2003, Stein1970}. It is widely used and has applications in various disciplines~\cite{Ohtake:2003,Melenk1996}. The approach involves a collection of \emph{patches} $\Pi = \{\Pi_i\}$ forming a locally-finite open cover of a given domain $\Omega \subseteq \RE^d$. The partition of unity is a set of non-negative smooth \emph{partition functions} $\{\phi_i\}$ such that the support of $\phi_i$, denoted $\supp(\phi_i)$, is a subset of $\Pi_i$ (see Figure~\ref{fig:cover}(a)). The name derives from the requirement that for all $x \in \Omega$, $\sum_i \phi_i(x) = 1$.

\begin{figure}[htbp]
    \centering\includegraphics[scale=0.40]{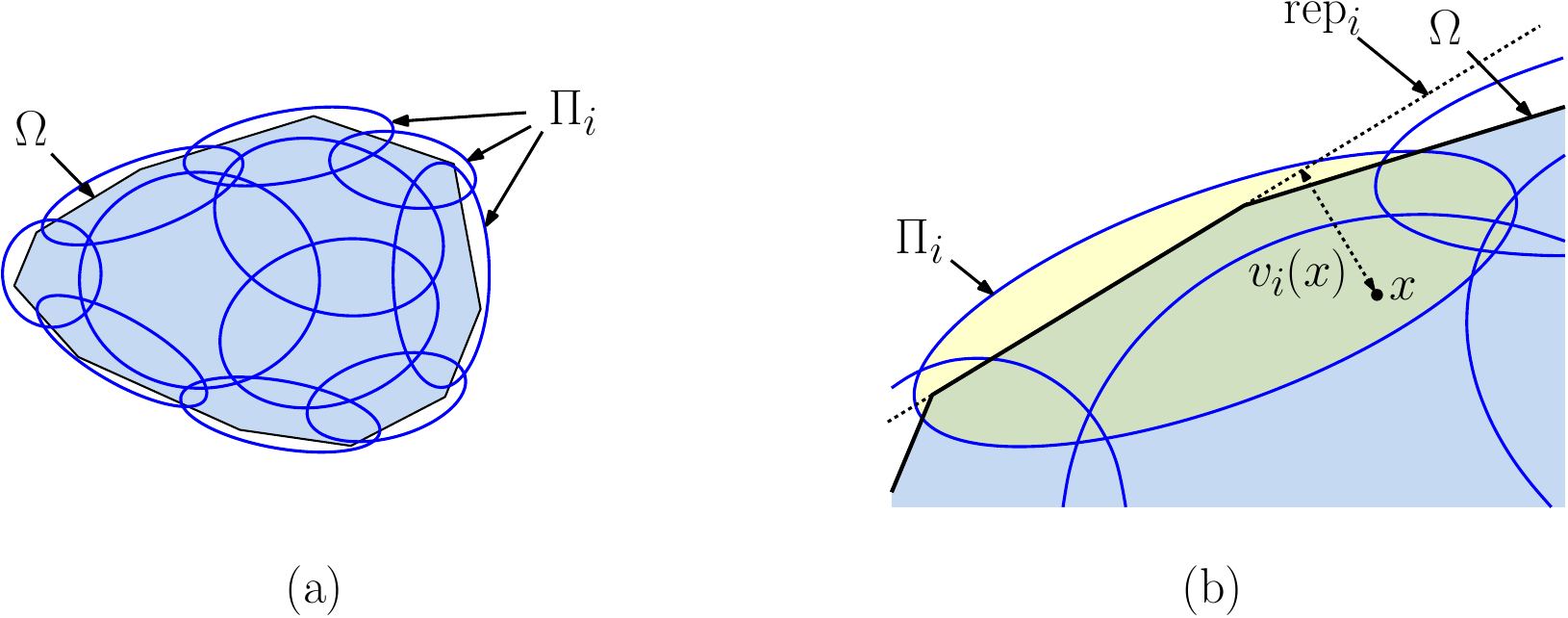}
    \caption{Patches, representatives, and the partition of unity.}
    \label{fig:cover}
\end{figure}

In the context of distance approximation, let us assume that each patch is associated with a \emph{local distance function} $v_i$, such that the restriction of $v_i$ to $\Pi_i$ is an $\eps$-approximation to the true distance function $d_{\Omega}$. Concretely, each patch is associated with a \emph{representative}, denoted $\rep_i$. For example, when approximating the distance to a discrete point set $P$, $\rep_i$ may be a point $p \in P$. In our case, where $\Omega$ is a convex polytope, $\rep_i$ will be chosen to be a supporting hyperplane of a facet of $\Omega$ (see Figure~\ref{fig:cover}(b)). Then, $v_i(x)$ can be defined to be distance from $x$ to the associated representative,
\begin{equation}\label{eq:v_i}
    v_i(x) ~ = ~ \dist(x, \rep_i).
\end{equation}
The final approximate distance map results by taking the sum of these local distance functions over all patches weighted by the associated partition functions. 
\begin{equation}\label{eq:pu-total}
    \adm_{\Omega}(x) 
        ~ = ~ \sum_i \phi_i(x) \cdot v_i(x).
\end{equation}
Recall that the support of $\phi_i$ is limited to $\Pi_i$, so we need only compute the sum over patches containing $x$. Define the \emph{depth} of $x$ with respect to $\Pi$, denoted $\depth{\Pi}(x)$, to be the number of patches of $\Pi$ containing $x$, and define $\depth{\Pi} = \max_x \depth{\Pi}(x)$. As in standard applications of the partition-of-unity method, we will design our patches so that $\depth{\Pi}$ is $O(1)$.

In order to enforce the condition that the functions $\phi_i$ sum to unity at any point in the domain, we will define a set of smooth, non-negative \emph{weight functions} $\{\psi_i\}$, and then define
\begin{equation}\label{eq:pu-weight}
    \phi_i(x) 
        ~ = ~ \frac{\psi_i(x)}{\Psi(x)}, \text{ where } \Psi(x) = \sum_i \psi_i(x).
\end{equation}

Observe that since $\adm_{\Omega}(x)$ is a convex linear combination of functions, each of which is locally an $\eps$-approximate distance map for $\Omega$, it follows that $\adm_{\Omega}(x)$ is itself an $\eps$-approximate distance map. Our construction will guarantee that there exists a positive constant $\Psi_{\min}$, such that $\Psi(x) > \Psi_{\min}$, for all $x \in \Omega$ (see Lemma~\ref{lem:positive_Psi} in Section~\ref{sec:pu-math}). It follows that $\phi_i(x)$ can be made as smooth as desired, being the quotient of two positive continuous functions. Assuming that the local distance approximations $\{v_i\}$ are smooth, it follows that $\adm_{\Omega}$ is itself smooth, being a sum of products of pairs of continuous functions. As a 1-dimensional example, see Figure~\ref{fig:blending}.

\begin{figure}[htbp]
    \centering\includegraphics[scale=0.40]{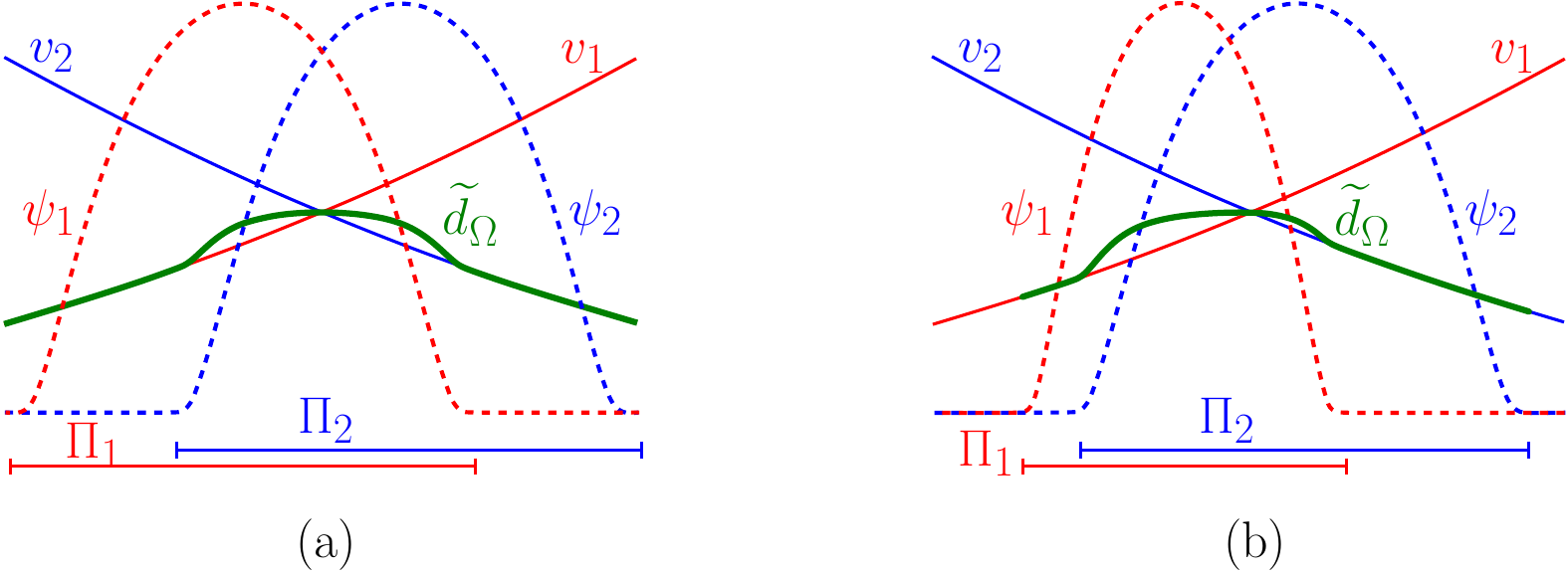}
    \caption{Blending two distance functions $\{v_i\}$ using two overlapping intervals $\{\Pi_i\}$ with associated weight functions $\{\psi_i\}$, yielding a smooth approximation $\adm_{\Omega}$ using (a) symmetric covers and (b) non-symmetric covers.}
    \label{fig:blending}
\end{figure}

It remains to define the weight function $\psi_i$ associated with each patch. These functions depend on the patch's shape. For our application, patches will be ellipsoids, but for this introduction, let us consider the simple case of a Euclidean ball with center point $c_i$ and radius $r_i$. First, for $x \in \RE^d$, define
\[
    f_i(x) 
        ~ = ~ \frac{1}{r_i^2}\|x - c_i\|^2.
\]
Observe that $f$ achieves its minimum value of $0$ at the ball's center and grows to $1$ at its boundary. To obtain a compactly-supported weight function, we use the standard technique of composing $f$ with a bump function, also known as the \emph{standard mollifier}~\cite{Showalter2011}
\begin{equation} \label{eq:bump}
    \mu(\sigma) 
        ~ = ~ \begin{cases}
                \exp\left(- \dfrac{1}{1 - \sigma^2}\right)    & \text{if $|\sigma| < 1$},\\
                0                                           & \text{otherwise}.
            \end{cases}
\end{equation}
Since $\mu(0) = e^{-1}$ and $\mu(1) = 0$, we see that the weight is highest near the middle of the shape, where $f = 0$, and decays gracefully towards the boundary, where $f = 1$. It is well-known that $\mu \in C_c^\infty(\RE)$ and is non-analytic with vanishing derivatives for $|\sigma| = 1$~\cite{Showalter2011}. Therefore, we may define $\psi_i(x) = \mu(f_i(x))$. 

In summary, given any query point $x$, we first determine the patches that contain it. (The number of which, $\depth{\Pi}(x)$, will be bounded by a constant.) Given the shape functions $f_i$ for each of these patches, we compute the weight functions $\psi_i$'s by applying the mollifier of Eq.~\eqref{eq:bump}. We then apply Eq.~\eqref{eq:pu-weight} to obtain the partition-of-unity blending functions. Finally, we apply Eqs.~\eqref{eq:v_i} and \eqref{eq:pu-total} to obtain the final smooth distance approximation. The overall space and query time are dominated by the total number of patches and the time needed to determine which patches contain the query point, respectively.

\section{Approximating the Boundary Distance Function} \label{sec:medial-apx}

The process described in the previous section is generic and can be applied in settings where the answer to the query can be expressed in terms of a covering of space by regions of low combinatorial complexity. For the sake of illustration, let us now explore how this works in the specific application of computing a smooth absolute $\eps$-approximation to the boundary distance in a convex polytope $\Omega$ in $\RE^d$. Let us assume that $\Omega$ has been scaled uniformly to unit diameter, so the absolute approximation error is $\eps$. 

We will employ a standard method for reducing distance approximation to a covering problem. First, let's consider the \emph{graph} of the boundary distance function $\dm_{\Omega}$, that is, the manifold $(x, \dm_{\Omega}(x))$ in $\RE^{d+1}$. We will use $z$ to denote coordinate values along the $(d+1)$st coordinate axis, which we will take to be directed vertically upwards in our drawings (see Figure~\ref{fig:medial-lift}(b) and (c)).

\begin{figure}[htbp]
    \centering\includegraphics[scale=0.40]{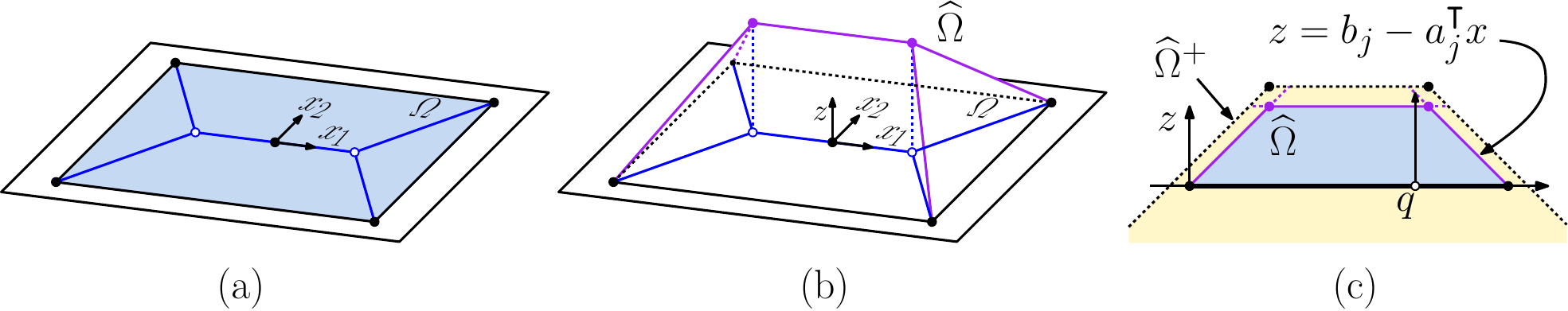}
    \caption{Lifting the polytope $\Omega \subseteq \RE^d$ to the lifted body $\lift{\Omega} \subseteq \RE^{d+1}$.}
    \label{fig:medial-lift}
\end{figure}

Assuming that the polytope $\Omega$ contains the origin in its interior, it can be represented as the intersection of a set of $n$ halfspaces in $\RE^d$, $\Omega = \bigcap_{j=1}^n H_j$, with each $H_j$ taking the form
\[
    H_j 
        ~ = ~ \{x \in \RE^d \ST a_j^{\Transpose} x \leq b_j \},
\]
where $a_j \in \RE^d$ is an outward-pointing unit normal vector orthogonal to $H_j$'s bounding hyperplane and $b_j \in \RE^+$ is the distance of the bounding hyperplane from the origin. The distance of a point $x \in \Omega$ to the bounding hyperplane is the non-negative scalar $z$ such that $x + z a_j$ lies on the bounding hyperplane, that is $z = b_j - a_j^{\Transpose} x$. The set of points lying below this surface (that is, the hypograph of the distance function) is the halfspace in $\RE^{d+1}$ given by the linear inequality $z \leq b_j - a_j^{\Transpose} x$. The boundary distance function is just the lower envelope (or minimization diagram~\cite{HaK15}) of this set of halfspaces. To turn this into a bounded convex polytope, we add a horizontal ground-surface halfspace $\lift{H}_0 = \{(x; z) : z \geq 0\}$. Define the \emph{lifted body} $\lift{\Omega} \subset \RE^{d+1}$ to be
\[
    \lift{\Omega}
        ~ = ~ \lift{H}_0 \cap \bigcap_{j=1}^n \lift{H}_j,
        ~~~\text{where $\lift{H}_j = \{ (x; z) \in \RE^{d+1} \ST z \leq b_j - a_j^{\Transpose} x\}$} \text{ for } j \in [n].
\]
Since the sides have a slope of $+1$, the diameter of $\lift{\Omega}$ is $O(1)$.

To achieve an absolute approximation error of at most $\eps$, we lift each of the upper halfspaces of $\lift{\Omega}$ by a vertical distance of $+\eps$ to obtain the resulting expanded object $\lift{\Omega}^+$. That is, we define $\lift{H}^{\delta}_j = \{ (x; z) \in \RE^{d+1} \ST z \leq b_j - a_j^{\Transpose} x + \delta \}$ and $\lift{\Omega}^+ = \bigcap_{j=1}^n \lift{H}^{\eps}_j$ (see Figure~\ref{fig:medial-lift}(c)). Note that the ground-surface halfspace ($\lift{H}_0$) is not needed for $\lift{\Omega}^+$, and hence it is unbounded. The essential features of lifting and expansion are summarized in the following lemma.

\begin{lemma} \label{lem:lift-summary}
Given a convex polytope $\Omega$ of unit diameter and any $x \in \Omega$, if a vertical ray is shot upwards from $x$ (viewed as a point in $\RE^{d+1}$) hits a bounding hyperplane of $\lift{\Omega}$ within $\lift{\Omega}^+$, then the associated facet of $\Omega$ is an absolute $\eps$-approximate nearest neighbor of $x$.
\end{lemma}

The upshot is that we can base the local distance functions $v_i(x)$ (recall Eq.~\eqref{eq:v_i}) on the distance to the bounding hyperplane of $\Omega$ corresponding to the bounding hyperplane in the lifted body $\lift{\Omega}$ that is hit by the vertical ray shot upwards from the query point $x$. An important feature of $\lift{\Omega}^+$, which will be of later use (in Lemma~\ref{lem:ball-containment}), is that the distance between its boundary and that of $\lift{\Omega}$ is at least $c \kern+1pt \eps \cdot \diam(\lift{\Omega})$, for some constant $c$.

\subsection{Macbeath Regions and Ellipsoids} \label{sec:macbeath}

Our approach to approximating $\lift{\Omega}$ for the purpose of answering distance map queries will be based on generating a net-like covering of $\lift{\Omega}$ based on objects called \emph{Macbeath regions}. Macbeath regions and their variants have been widely used in convex approximation (see, e.g., \cite{AFM17a,AFM17c,AbM18,AAFM22,AFM23}). In contrast to traditional covers based on subdivisions by fat objects, e.g., hypercubes, Macbeath regions naturally adapt to the shape of the object being covered. In this section we present a brief review of the salient features of Macbeath regions.

Given a convex body $\Omega$ and any point $x \in \Omega$, the \emph{Macbeath region} at $x$ is the largest centrally-symmetric body centered at $x$ and contained within $\Omega$. It is common to apply a constant scaling factor. Formally, for $\lambda \in \RE^+$, the \emph{$\lambda$-scaled Macbeath region} at $x$ is
\[
    M_\Omega^{\lambda}(x) 
        ~ = ~ x + \lambda ((\Omega-x) \cap (x-\Omega))
\]
(see Figure~\ref{macbeath.fig}(a)). When $\Omega$ is clear from context, we will often omit the subscript. We refer to $x$ and $\lambda$ as the \emph{center} and \emph{scaling factor} of $M^{\lambda}(x)$, respectively. When $\lambda < 1$, we say $M^{\lambda}(x)$ is \emph{shrunken}.

\begin{figure}[htbp]
  \centerline{\includegraphics[scale=0.40]{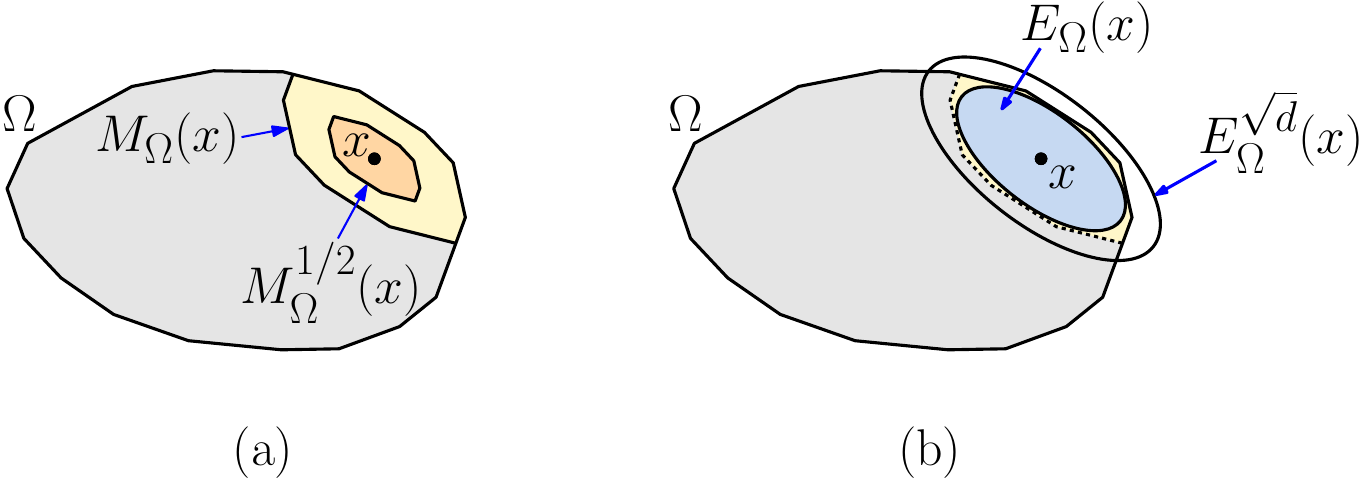}}
  \caption{(a) Macbeath regions and (b) Macbeath ellipsoids.} \label{macbeath.fig}
\end{figure}

It is useful to have a low-complexity, smooth proxy for a Macbeath region. Given a Macbeath region, define its associated \emph{Macbeath ellipsoid} $E^{\lambda}_\Omega(x)$ to be the maximum-volume ellipsoid contained within $M^{\lambda}_\Omega(x)$ (see Figure~\ref{macbeath.fig}(b)). Clearly, this ellipsoid is centered at $x$, and $E^{\lambda}_\Omega(x)$ is a $\lambda$-factor scaling of $E^1_\Omega(x)$ about $x$. By John's Theorem~\cite{Bal97}
\[
  E^{\lambda}_{\Omega}(x)
	~ \subseteq ~ M^{\lambda}_{\Omega}(x)
	~ \subseteq ~ E^{\lambda\sqrt{d}}_{\Omega}(x).
\]
 Chazelle and Matou\v{s}ek showed that this ellipsoid can be computed for a convex polytope $\Omega$ in time linear in the number of its bounding halfspaces \cite{ChM96}. 
 
 A fundamental property of Macbeath regions, called \emph{expansion-containment}, states that if two shrunken Macbeath regions (or ellipsoids) overlap, then a constant-factor expansion of one contains the other. There are many formulations. The following can be found in~\cite{AbM18}.

\begin{lemma}[Expansion-Containment] \label{lem:exp-con}
Given a convex body $\Omega \in \RE^d$, $0 < \lambda < 1$, let $\beta = (3+\lambda)/(1 - \lambda)$. Then for any $x, y \in \Omega$:
\begin{description}
\item[$(i)$] $M^{\lambda}(x) \cap M^{\lambda}(y) \neq \emptyset ~\Longrightarrow~ M^{\lambda}(y) \subseteq M^{\beta\lambda}(x)$,

\item[$(ii)$] $E^{\lambda}(x) \cap E^{\lambda}(y) \neq \emptyset ~\Longrightarrow~ E^{\lambda}(y) \subseteq E^{\beta\lambda\sqrt{d}}(x)$.
\end{description}
\end{lemma}

\subsection{Approximation through Covering} \label{sec:cover}

Our approach to computing an $\eps$-approximation to the boundary distance function within a convex polytope $\Omega$ in $\RE^d$ utilizes Macbeath regions to cover the lifted body $\lift{\Omega}$ in $\RE^{d+1}$. Recall our assumption that $\Omega$ has been scaled to unit diameter, and hence the lifted body $\lift{\Omega}$ also has unit diameter. Given this scaling, our objective is to answer vertical ray-shooting queries up to an absolute error of at most $\eps$, (see Lemma~\ref{lem:lift-summary}).

Before presenting our solution, let us recall some known results for convex approximation. Given a convex body $\Omega$ of unit diameter and $\eps > 0$, an \emph{$\eps$-approximate polytope membership query} is given a query point $q$ and returns positive answer if $q$ lies within $\Omega$, a negative answer if $q$ lies at distance more than $\eps$ from $\Omega$, and otherwise, it may give either answer. Arya {\etal} presented an efficient data structure for answering approximate membership queries~\cite{AFM17a}. Later, Abdelkader and Mount~\cite{AbM18} presented a simpler approach with the same space and query times, as described in the following lemma. We will employ a variant of the latter data structure.

\begin{lemma} \label{lem:apx-membership}
Given a convex polytope $\Omega \in \RE^d$, there exists a data structure that can answer absolute $\eps$-approximate polytope membership queries for $\Omega$ in time $O(\log (1/\eps))$ and storage $O(1/\eps^{(d-1)/2})$.
\end{lemma}

In order to apply this data structure for our purposes, we will need to delve a bit deeper into how it works. Our application of this structure will be in the lifted space $\RE^{d+1}$, but let us describe it now for an arbitrary convex body $\Omega$ in $\RE^d$. Given a non-negative parameter $\delta$, define the \emph{expanded body} $\Omega_{\delta}$ to be a convex set such that $\Omega \subseteq \Omega_{\delta}$, and the minimum distance between their boundaries of is at least $\delta$.

Next, we define the notion of a \emph{Macbeath-based Delone set} of $\Omega$ relative to $\Omega_{\delta}$. This structure is parameterized by two constants $0 < \lambda_p < \lambda_c < 1$, called the \emph{packing} and \emph{covering} constants, respectively (which may depend on the dimension $d$). Given any point $x \in \Omega$, define the \emph{covering ellipsoid} $E'_{\delta}(x) = E_{\Omega_{\delta}}^{\lambda_c}(x)$, that is, a Macbeath ellipsoid centered at $x$ with scaling factor $\lambda_c$ defined with respect to the expanded body $\Omega_{\delta}$. Define the \emph{packing ellipsoid} $E''_{\delta}(x)$ analogously, but with a scaling factor of $\lambda_p$. A \emph{Macbeath-based Delone set} for $\Omega$ relative to $\Omega_{\delta}$ is any maximal set of points $X \subset \Omega$, such that the packing ellipsoids $E''_{\delta}(x)$ centered at these points are pairwise disjoint. Abdelkader and Mount~\cite{AbM18} showed that, by standard properties of Macbeath regions, constants $\lambda_p$ and $\lambda_c$ can be chosen such that $X$ has the following properties: 
\begin{description}
\item[$(i)$] The union of the covering ellipsoids $E'_{\delta}(x)$ over all $x \in X$ covers the original body, $\Omega$,
\item[$(ii)$] For each $x \in X$, $E'_{\delta}(x)$ is contained within the expanded body, $\Omega_{\delta}$,
\item[$(iii)$] The number of ellipsoids $E'_{\delta}$ that contain any point $x$ is $O(1)$,
\item[$(iv)$] $|X| = O(1/\delta^{(d-1)/2})$.
\end{description}
Note that the constant factors hidden in the $O$-notation depend on the dimension $d$.

To turn this into an approximate search structure, a layered DAG is constructed as follows. For $i \geq 0$, let $\delta_i = 2^i \delta$. Construct a series of such Delone sets, $X_0, X_1, \ldots, X_m$ where $X_i$ is any Macbeath-based Delone set of $\Omega$ with respect to $\Omega_{\delta_i}$. As $i$ increases, the expanded body grows larger, and hence the Macbeath ellipsoids also grow larger. But since they need only cover the original body $\Omega$, their size, $|X_i|$, decreases with $i$. The final layer $\ell_m$ is defined to be the smallest integer such that $|X_{\ell_m}| = 1$. (It can be shown that $\ell_m = O(\diam(\Omega)) = O(1)$, which implies that $m = O(\log (1/\delta))$.) The leaves of the DAG correspond to the covering ellipsoids $E'_{\delta_0}$ centered at the points of $X_0$. The root corresponds to the covering ellipsoid $E'_{\delta_{\ell_m}}$ associated with the single point of $X_{\ell_m}$ (see Figure~\ref{fig:hierarchy-ellipsoid}). Finally, the nodes of level $i$ are connected to nodes at level $i-1$ whenever their associated $E'$ ellipsoids overlap. Abdelkader and Mount~\cite{AbM18} showed the following:
\begin{itemize}
\item The DAG has $O(\log (1/\eps))$ layers.
\item The out-degree of any node in the DAG is $O(1)$.
\item The total number of nodes in the DAG is $O(1/\delta^{(d-1)/2})$.
\end{itemize}
Lemma~\ref{lem:apx-membership} follows by applying a natural search process which simply descends from the root to any leaf in the DAG, always visiting a node whose $E'$ ellipsoid contains the query point. If a query point $q \in \Omega$, the search will succeed in finding a leaf-level ellipsoid $E'$ that contains this point.

\begin{figure}[htbp]
    \centering\includegraphics[scale=0.40]{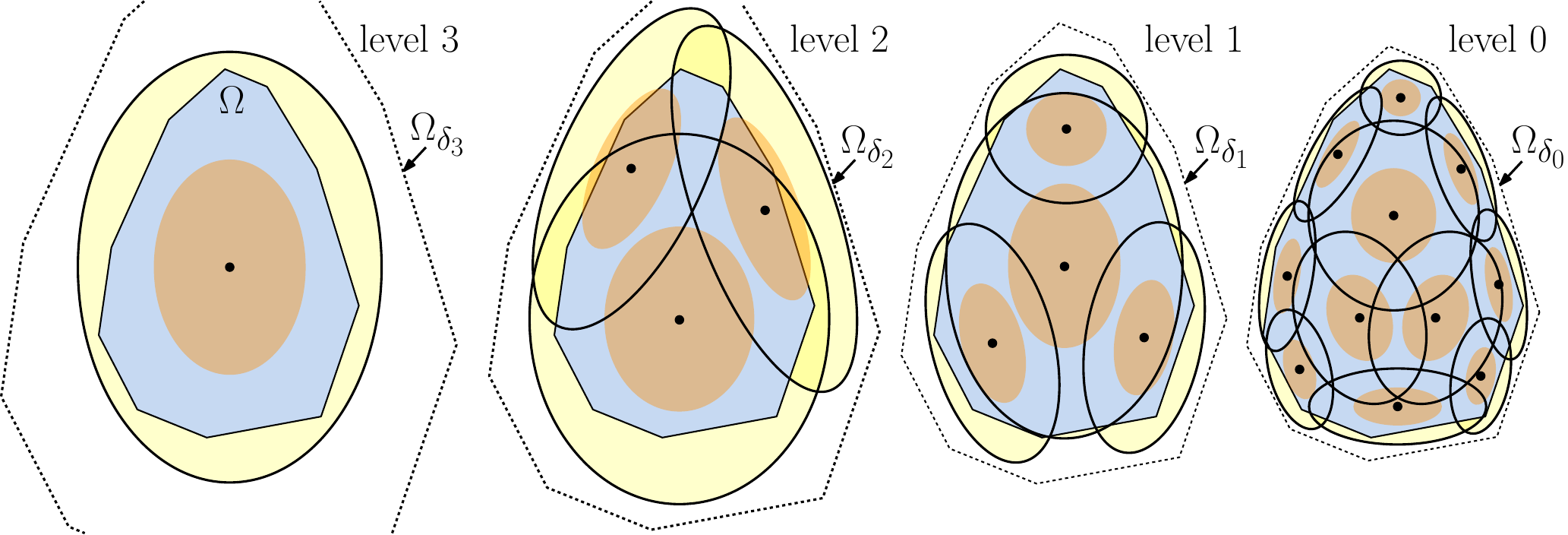}
    \caption{Layers of the Macbeath-based Delone set data structure.}
    \label{fig:hierarchy-ellipsoid}
\end{figure}

\subsection{Approximation through Vertical Ray Shooting} \label{sec:ray-shoot}

We can now explain how to construct the smooth boundary distance approximation for $\Omega$. First, we construct the lifted body $\lift{\Omega}$, as described in Section~\ref{sec:medial-apx}. Given a query point $q \in \Omega$, its distance to the boundary is determined by the height of the point on $\bd \kern+1pt \lift{\Omega}$ hit by an upward-directed vertical ray shot from $q$ in $\RE^{d+1}$. To apply the hierarchical search, let $\delta = \eps$, and for any $\ell \geq 0$, define $\lift{\Omega}_{\delta_{\ell}}$ to be the unbounded convex set that results by translating all the upper halfspaces bounding $\lift{\Omega}$ up by distance $\delta_{\ell} = 2^{\ell} \eps$. That is, $\lift{\Omega}_{\delta_{\ell}} = \bigcap_{j=1}^n \lift{H}^{\delta_{\ell}}_j$. Observe that $\lift{\Omega}_{\delta_0} = \lift{\Omega}^+$, the $\eps$-expanded body.

We can apply the data structure described in Lemma~\ref{lem:apx-membership} to these bodies in $\RE^{d+1}$. For any level of the structure, we say that a covering ellipsoid is a \emph{top ellipsoid} if there exists $x \in \Omega$, such that this ellipsoid has the highest intersection point with the vertical ray directed up from $x$ among all the ellipsoids in the cover. Because the covering ellipsoids cover $\lift{\Omega}$ it follows that the union of the top ellipsoids, when projected vertically onto $\RE^d$, covers the original body $\Omega$. Later on, the vertical projections of the top ellipsoids of the leaf level will serve as the covering patches $\{\Pi_i\}$ for the purposes of blending (see Figure~\ref{fig:cover}).

To answer vertical ray-shooting queries for point $q$ using our hierarchy, we traverse the hierarchy of ellipsoids, but whenever we descend a level in the DAG structure, among all the child nodes whose covering ellipsoid $E'$ intersects the vertical ray passing through $q$, we visit the one having the highest point of intersection with the ray (see Figure~\ref{fig:ray-shooting}). It was shown in~\cite{AbM18} that the number of ellipsoids that need to be considered is $O(1)$. Therefore, in time proportional to the number of levels, which is $O(\log (1/\eps))$, we can find the top ellipsoid at the leaf level traversed by the vertical ray.

\begin{figure}[htbp]
    \centering\includegraphics[scale=0.40]{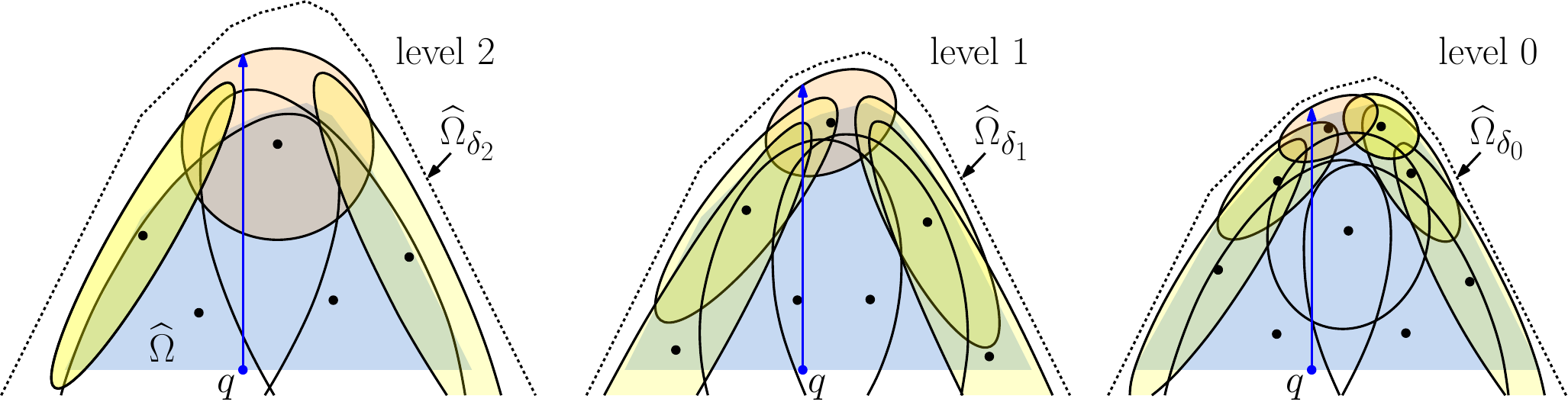}
    \caption{Vertical ray-shooting in the hierarchical structure.}
    \label{fig:ray-shooting}
\end{figure}

Furthermore, because all these ellipsoids lie within $\lift{\Omega}_{\delta_0}$, the terminus of the vertical ray lies within the vertical gap of length $\eps$ between $\lift{\Omega}$ and $\lift{\Omega}^+$, as required in Lemma~\ref{lem:lift-summary}. Finally, through an appropriate adjustment of the scaling factors $\lambda_p$ and $\lambda_c$, we can apply the same analysis as in Arya {\etal}~\cite[Lemma 3.5]{AAFM22} to find a witness hyperplane that serves as the representative for all vertical rays passing through this ellipsoid. Since the construction is performed in $\RE^{d+1}$, the space required is $O(1/\eps^{d/2})$. This implies that, even ignoring continuity, we can answer $\eps$-approximate boundary distance queries efficiently.

\begin{theorem} \label{thm:apx-lift-membership}
Given a convex polytope $\Omega \in \RE^d$, there exists a data structure that can answer absolute $\eps$-approximate boundary distance queries (without continuity guarantees) for $\Omega$ in time $O(\log (1/\eps))$ and storage $O(1/\eps^{d/2})$.
\end{theorem}

\subsection{Additional Properties} \label{sec:constructs_for_blending}
There are a couple of additional properties, which will be useful for the task of computing smooth distance approximations. First, because the boundaries of $\lift{\Omega}$ and $\lift{\Omega}^+$ are separated by a vertical distance of $\eps$, we can infer that the Macbeath ellipsoids cannot be too skinny.

\begin{restatable}{lemma}{BallContainment}\label{lem:ball-containment}
Each of the covering Macbeath ellipsoids in the data structure of Theorem~\ref{thm:apx-lift-membership} contains a Euclidean ball at its center of radius at least $c \kern+1pt \eps$, for some constant $c$ (depending on $\lambda_c$ and $d$).
\end{restatable}

\begin{proof}
All covering ellipsoids lie within $\lift{\Omega}$, and their Macbeath regions are defined with respect to $\lift{\Omega}^+$. Recall that the boundary of $\lift{\Omega}^+$ has a vertical separation of $\eps$ from $\lift{\Omega}$, and the facets of both bodies have a slope (with respect to vertical) of at most $1$ (see Figure~\ref{fig:medial-lift}). Hence, the distance from any point in $\lift{\Omega}$ to the boundary of $\lift{\Omega}^+$ is at least $\eps/\sqrt{d}$. Therefore, the Macbeath region contains a ball of this radius at its center. Scaling by a factor of $\lambda_c$, the covering Macbeath region has a ball of radius $\eps(\lambda_c/\sqrt{d})$. By John's theorem, the covering ellipsoid contains a ball of radius $\eps(\lambda_c/\sqrt{d})/\sqrt{d} = \eps(\lambda_c/d)$. Setting $c = \lambda_c/d$, completes the proof.
\end{proof}

Second, from the properties of the Macbeath-based Delone set, each point of $\lift{\Omega}$ is covered by only a constant number of covering ellipsoids at the leaf level. While this does not necessarily hold for the vertical projections of these ellipsoids, it does hold when we restrict attention to the top ellipsoids. Let $\depth{\Pi}$ denote maximum number of ellipsoids that may contain any point of $\Omega$.

\begin{restatable}{lemma}{PUDepth} \label{lem:PU_depth}
The blending patches $\Pi$ resulting from the vertical projections of the top covering ellipsoids have constant depth, that is, $\depth{\Pi} = O(1)$.
\end{restatable}

Before giving the proof, we establish a useful technical lemma. Let us begin with some notation. Given a concave function $f: \RE^d \to \RE$, its \emph{hypograph}, denoted $f^-$, is the set of points in $\RE^{d+1}$ lying on or below the function. Clearly, $f^-$ is a convex set. For any point $x \in f^-$, define the \emph{ray distance} of $x$ with respect to $f^-$, denoted $\ray_{f^-}(x)$ to be the length of a ray shot upwards from $x$ to the boundary of $f^-$. For any $x \in f^-$ and $\lambda \geq 0$, define $M^{\lambda}_{f^-}(x)$ as the $\lambda$-scaled Macbeath region relative to $f^-$. We omit explicit references to $f^-$ in subscripts when it is clear from context.

\begin{lemma} \label{lem:ray-dist}
Given concave $f: \RE^d \to \RE$, $x \in f^-$ and $\lambda \geq 0$, for all $y \in M^{\lambda}_{f^-}(x)$,
\[
    (1 - \lambda) \cdot \ray_{f^-}(x)
        ~ \leq ~ \ray_{f^-}(y)
        ~ \leq ~ (1 + \lambda) \cdot \ray_{f^-}(x).
\]
\end{lemma}

\begin{proof}
To simplify notation, let $r_x = \ray_{f^-}(x)$ and $r_y = \ray_{f^-}(y)$. To prove the upper bound, let $x'$ denote the intersection of the vertical ray through $x$ with the boundary of $f^-$ (see Figure~\ref{fig:ray-dist}(a)). Consider a supporting hyperplane $h_0$ for $f^-$ passing through $x'$. Let $h_1$ be the parallel supporting hyperplane passing through $x$, and let $h_2$ be the parallel supporting hyperplane along the lower side of $M^{\lambda}(x)$. Clearly, the vertical distance between $h_0$ and $h_1$ is $r_x$, and the vertical distance between $h_1$ and $h_2$ is $\lambda r_x$. Since $M^{\lambda}(x)$ lies entirely above $h_2$, it follows that the vertical segment defining $r_y$ lies entirely below $h_0$ and above $h_2$, which implies that $r_y \leq r_x + \lambda r_x = (1 + \lambda) r_x$, as desired.

\begin{figure}[htbp]
    \centering\includegraphics[scale=0.40]{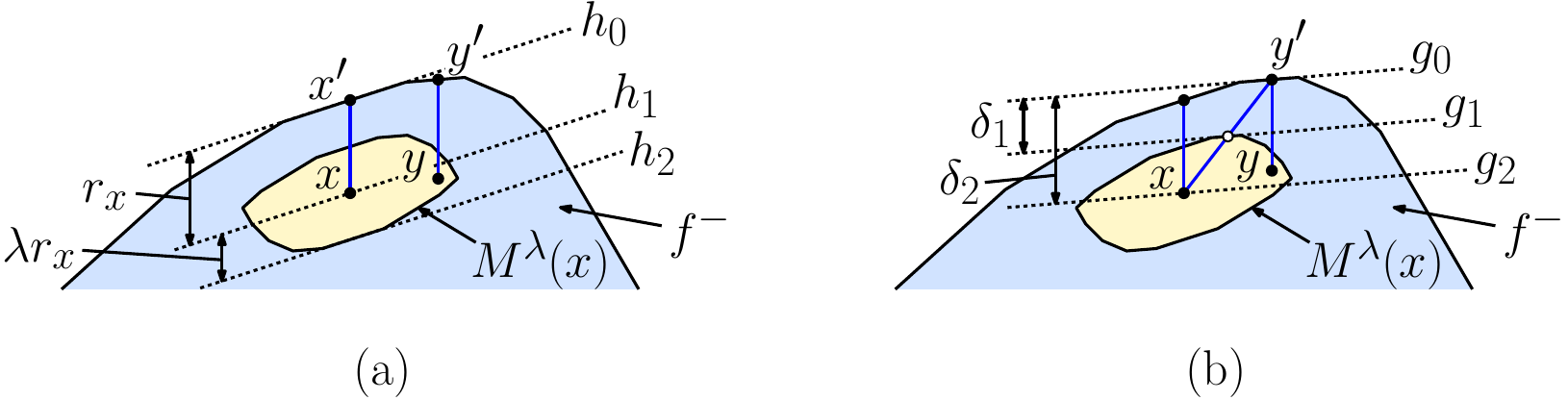}
    \caption{Proof of Lemma~\ref{lem:ray-dist}.}
    \label{fig:ray-dist}
\end{figure}

To prove the lower bound, let $y'$ denote the point where the vertical ray through $y$ intersects the boundary of $f^-$ (see Figure~\ref{fig:ray-dist}(b)). Let $g_0$ denote a supporting hyperplane for $f^-$ passing through $y'$. Let $g_1$ be the upper parallel supporting hyperplane for $M^{\lambda}(x)$, and let $g_2$ be the parallel hyperplane passing through $x$. Let $\delta_1$ denote the vertical distance between $g_0$ and $g_1$, and let $\delta_2$ denote the vertical distance between $g_0$ and $g_2$. By definition of the Macbeath region, we have $\delta_2 - \delta_1 = \lambda \delta_2$, or equivalently $\delta_1 = (1 - \lambda) \delta_2$. Clearly, $y$ lies below $g_1$, and so $r_y \geq \delta_1$. Since all of $f^-$ (including $x'$) lies below $g_0$, we have $r_x \leq \delta_2$. Therefore, $r_y \geq \delta_1 = (1-\lambda) \delta_2 \geq (1 - \lambda) r_x$, as desired.
\end{proof}

\begin{proof} (Of Lemma~\ref{lem:PU_depth}.)

Recall that constants $\lambda_c$ and $\lambda_p$ are the so called \emph{covering} and \emph{packing} scale factors used in our construction. Given a point $x \in \Omega^+$, let $M'(x)$ and $M''(x)$ denote respectively the covering ($\lambda_c$-scaled) and packing ($\lambda_p$-scaled) Macbeath regions centered at $x$ with respect to the expanded body $\lift{\Omega}^+$. Define $E'(x)$ and $E''(x)$ analogously for Macbeath ellipsoids. 

Recall that our construction is based on a maximal point set $X \subset \lift{\Omega}$ such that the packing ellipsoids $E''(x)$ are disjoint for $x \in X$ and $E'(x)$ cover $\lift{\Omega}$. Given any $q \in \Omega$, let $X(q) \subseteq X$ denote the set of top covering ellipsoids $E'(x)$ whose vertical projection contains $q$. Equivalently, $x \in X(q)$ if the vertical ray passing through $q$ intersects $E'(x)$. It suffices to show that for any $q \in \Omega$, $|X(q)| = O(1)$.

For any $x \in \lift{\Omega}^+$, define $\ray(x)$ to be the length of a vertical ray shot from $x$ up to the boundary of $\lift{\Omega}^+$. By definition of a top ellipsoid, for any $x \in X(q)$, there exists a point $z \in E'(x)$ such that $\ray(z) \leq \eps$. Since $E'(x) \subseteq M'(x)$, we have $z \in M'(x)$ (see Figure~\ref{fig:pu-depth}(a)). Thus, by Lemma~\ref{lem:ray-dist} with $\lift{\Omega}^+$ playing the role of $f^-$ and $z$ playing the role of $y$, it follows that $\ray(x) \leq \ray(z)/(1-\lambda_c)  \leq \eps/(1-\lambda_c)$. Applying the lemma again, it follows that for any other point $y \in E'(x)$, we have 
\[
    \ray(y) 
            ~ \leq ~ (1+\lambda_c) \cdot \ray(x)
            ~ \leq ~ \frac{1+\lambda_c}{1-\lambda_c} \kern+1pt \eps.
\]
Also, because $x \in \lift{\Omega}$, we have $\ray(x) \geq \eps$, implying again by Lemma~\ref{lem:ray-dist} that $\ray(y) \geq (1-\lambda_c) \eps$. In summary, for each $x \in X(q)$, there exists a point $y$ along the vertical ray shot up from $q$ such that $(1-\lambda_c) \eps \leq \ray(y) \leq \frac{1+\lambda_c}{1-\lambda_c} \kern+1pt \eps$.

\begin{figure}[htbp]
    \centering\includegraphics[scale=0.40]{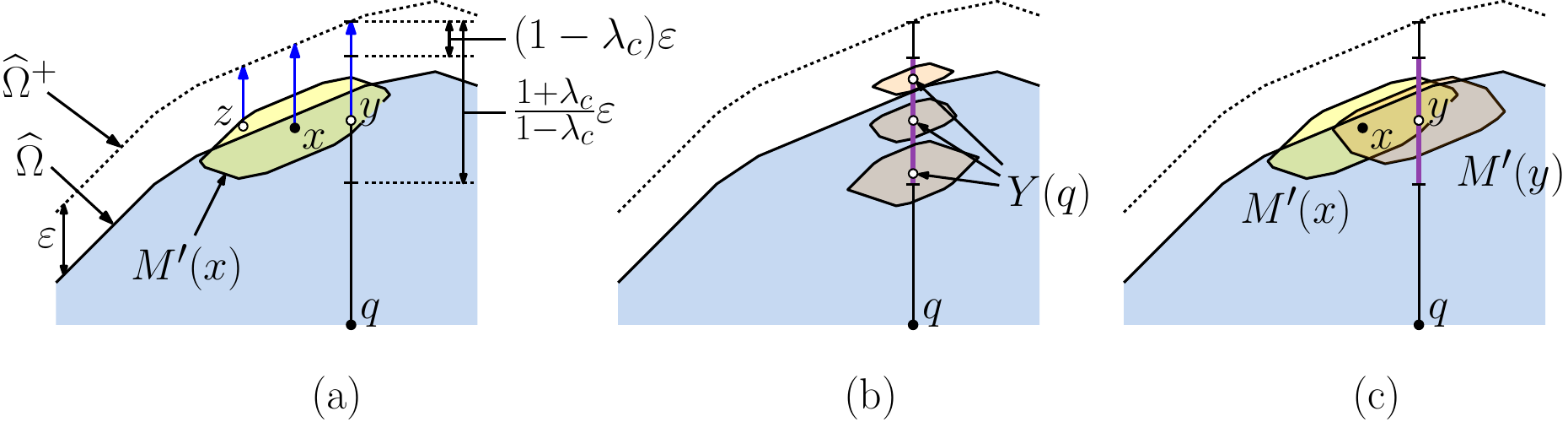}
    \caption{Proof of Lemma~\ref{lem:PU_depth}.}
    \label{fig:pu-depth}
\end{figure}

Let $Y(q)$ be any maximal set of points along the vertical line through $q$ that have ray distances in the interval $\big[(1-\lambda_c), \frac{1+\lambda_c}{1-\lambda_c}\big] \eps$ and whose packing Macbeath regions are pairwise disjoint (see Figure~\ref{fig:pu-depth}(b)). Each such Macbeath region covers an interval of length at least $\lambda_p (1 - \lambda_c) \eps$. By a standard packing argument, there are at most a constant $c'$ (depending on $\lambda_p$ and $\lambda_c$) of such Macbeath regions, and their covering Macbeath regions cover this subsegment of the vertical line. 

Now, associate each point $x \in X(q)$ with any one of the points of $y \in Y(q)$, such that $M'(x) \cap M'(y) \neq \emptyset$ (see Figure~\ref{fig:pu-depth}(c)). By the prior observations, such a point of $y$ exists for each $x$. By expansion containment (Lemma~\ref{lem:exp-con}), a constant factor expansion of $M'(y)$ contains $M'(x)$ and vice versa. Therefore, the volumes of these bodies are equal up to constant factors (depending on $\lambda_c$ and the dimension $d$). Because $\lambda_p$ and $\lambda_c$ are both constants, the volumes of $M''(x)$ and $M''(y)$ are related by constant factors. Thus, by a straightforward packing argument, the disjointness of the $M''(x)$ Macbeath regions implies that the number of $x \in X(q)$ that are associated with any $y \in Y(q)$ is bounded above by some constant $c''$. Thus, we have $|X(q)| \leq c' c'' = O(1)$, as desired.
\end{proof}

\section{Putting it Together} \label{sec:together}

We can now explain how to combine the results of the previous section with the partition-of-unity method from Section~\ref{sec:pou}, to obtain the final smooth distance approximation. 

The set of patches $\Pi = \{\Pi_i\}$ used in blending consist of the vertical projections of all the top ellipsoids from level-0 of the vertical ray-shooting data structure. Each ellipsoidal patch $\Pi_i$ is represented by its center $c_i \in \RE^d$ and a positive-definite matrix $M_i$ such that
\begin{equation}\label{eq:shapes}
    \Pi_i ~ = ~ \{ x \in \RE^d : f_i(x) \leq 1\}, 
        \qquad\text{where~} f_i(x) ~ = ~ (x - c_i)^{\intercal}M_i (x - c_i).
\end{equation}
Recalling the definition of the standard mollifier $\mu$ from Eq.~\eqref{eq:bump}, we define
\begin{equation}\label{eq:psi_i}
    \psi_i(x) ~ = ~ \mu(f_i(x)).
\end{equation}
Given these weight functions, we apply Eq.~\eqref{eq:pu-weight} to obtain the blending function $\phi_i(x)$ for each patch. Recall that each of the top ellipsoids $\Pi_i$ is associated with a representative of the upper envelope of $\lift{\Omega}$ in the form of a halfspace $H_i = \{x \in \RE^d \ST a_i^{\Transpose} x \leq b_i \}$. As mentioned in Section~\ref{sec:medial-apx}, the associated local distance function is $v_i(x) = b_i - a_i^{\Transpose} x$ (Eq.~\eqref{eq:v_i}).

Given a query point $q \in \Omega$, we use the vertical ray-shooting data structure to determine the patches $\Pi_i$ containing it. By Lemma~\ref{lem:PU_depth}, there are a constant number of them. We apply Eq.~\eqref{eq:pu-total} to blend together the local distance functions to obtain the final distance approximation, $\adm_{\Omega}(q)$. The space and query time are dominated by the complexity bounds for the ray-shooting data structure, given in Lemma~\ref{thm:apx-lift-membership}. This establishes the correctness and complexity bounds of Theorem~\ref{thm:main}. The bounds on the norms of the gradient and Hessian are presented in the next section.

\section{Partition of Unity Derivations} \label{sec:pu-math}

In this section, we derive various expressions of the gradients of the partition of unity and its constituent functions, culminating in an simple expression of the gradient of the distance approximation $\Gradient \adm_{\Omega}(x)$ along with an upper bound on its magnitude. Recall the we utilize a locally-finite open cover of a domain $\Omega \subseteq \RE^d$ by a collection of patches $\{\Pi_i\}$ with a partition of unity $\{\phi_i\}$ subordinate to the cover, i.e. $\supp(\phi_i) \subseteq \Pi_i$. Each patch $\Pi_i$ is associated with a local distance approximation $v_i(x) = b_i - a_i^{\Transpose} x$, where $(b_i, a_i)$ are the defining quantities for the hyperplane bounding the closest facet to the center point of the patch. The approximate distance function is the blend of these local distance functions
\[
    \adm_{\Omega}(x) 
        ~ = ~ \sum_i \phi_i(x) \cdot v_i(x).
\]

Recall the bump function $\mu: \RE \rightarrow \RE_{\geq 0}$ from Eq.~\eqref{eq:bump}, with $\supp(\mu) \subseteq [0, 1]$:
\[
    \mu(\sigma) 
        ~ = ~ \begin{cases}
                \exp\left(- \dfrac{1}{1 - \sigma^2}\right),   & |\sigma| < 1,\\
                0,                                            & \text{otherwise.}
            \end{cases}
\]
In what follows, the support of all functions is similarly bounded, and we omit explicit references to the complement of the support. Recall that we obtained the weight functions defining the partition of unity over a set of patches $\{\Pi_i\}$ using the following function as input to $\mu$:
\[
    f_i(x; c_i, M_i)   
          ~ = ~ (x - c_i)^{\intercal}M_i(x - c_i)
\]
The partition of unity is the set of functions $\bigg\{\phi_i = \dfrac{\psi_i}{\Psi}\bigg\}$, with $\Psi(x) = \sum_i \psi_i(x)$, such that $\psi_i(x) = \mu(f_i(x))$. The following lemma establishes a lower bound on $\Psi$.

\begin{lemma} \label{lem:positive_Psi}
For a suitable choice of the packing constant $\lambda_p$, the data structure described in Section~\ref{sec:medial-apx} ensures that $\Psi_{\min} > 1/4$.
\end{lemma}

\begin{proof}
Recall that the patches $\{\Pi_i\}$ arise as the projections of the top ellipsoids covering the lifted polytope $\lift{\Omega}$. We build the data structure in Section~\ref{sec:cover} by adjusting the packing constant $\lambda_p$ such that the union of the smaller covering ellipsoids $\bigcup_i \Pi^{1/2}_i$ covers the original body. As a result, any point $x \in \Omega$ is covered by some ellipsoid $\Pi^{1/2}_i$. It follows that $f_i(x) \leq 1/2$, ensuring $\psi_i(x) \geq \mu(1/2)$. We conclude that $\Psi(x) = \sum_j \psi_j(x) \geq \psi_i(x) \geq \mu(1/2) > 0.26$, which clearly exceeds $1/4$.
\end{proof}

\subsection{Blending Gradients}

For a given function $\sigma: \RE^d \rightarrow \RE$, we have
\[ 
    \Gradient \mu(\sigma) 
        ~ = ~ \exp\left(\frac{1}{\sigma^2 - 1}\right)\cdot\Gradient\left(1+\frac{1}{\sigma^2-1}\right) 
        ~ = ~ \mu(\sigma)\cdot\frac{-2\sigma}{(\sigma^2-1)^2}\cdot\Gradient \sigma.
\]
From the above equations, this yields $\Gradient f_i(x) = 2\cdot M_i(x - c_i)$. For weight functions, we have
\begin{align}
    \Gradient \psi_i(x) 
        & ~ = ~ \Gradient \mu(f_i(x)) 
          ~ = ~ \psi_i(x) \cdot \frac{-2f_i(x)}{(f_i(x)^2-1)^2} \cdot 2 M_i (x - c_i) \nonumber \\
        & ~ = ~ -4 \cdot \frac{\psi_i(x) f_i(x)}{(f_i(x)^2-1)^2} \cdot M_i(x - c_i). \label{eq:grad_wi_2}
\end{align}

Towards deriving the desired bound on the gradient of the distance approximation $\|\Gradient \adm_{\Omega}(x)\|$, we start by bounding the gradient of the weight functions $\|\Gradient \psi_i\|$.

\begin{lemma} \label{lem:blending_gradient_bound}
The data structure described in Section~\ref{sec:medial-apx} guarantees that both $\|\Gradient \psi_i(x)\|$ and $\|\Gradient \phi_i(x)\|$ are $O(1/\eps)$.
\end{lemma}

\begin{proof}
We begin by simplifying Eq.~\eqref{eq:grad_wi_2}. Graphing the function $\kappa(\sigma) = (\mu(\sigma) \cdot \sigma)/(\sigma^2 - 1)^2$ over the range $\sigma \in [-1,+1]$ reveals that its absolute value is strictly less than $1/2$. Since $|f_i(x)| < 1$ throughout the range of interest, it follows that
\begin{equation} \label{eq:kappa}
 \left| \frac{\psi_i(x) f_i(x)}{(f_i(x)^2-1)^2} \right|
    ~ = ~ | \kappa(f_i(x)) |
    ~ < ~ \frac{1}{2}.
\end{equation}

Hence, by Eq.~\eqref{eq:grad_wi_2} we have
\begin{align}
 \|\Gradient \psi_i(x)\| 
    & ~ = ~ \left\|-4 \psi_i(x) \cdot \frac{f_i(x)}{(f_i(x)^2-1)^2} \cdot M_i(x - c_i)\right\|
      ~ = ~ 4 \SP \kappa(f_i(x)) \cdot \| M_i(x - c_i) \| \nonumber\\
    & ~ < ~ 2 \cdot \| M_i (x - c_i) \|, \label{eq:grad_psi}
\end{align}
We can bound $\|M_i (x - c_i)\|$ as follows. If we express the matrix $M_i$ in terms of an orthonormal basis whose coordinate vectors are aligned with the ellipsoid's major axes, $M_i$ is a diagonal matrix whose entries are of the form $1/r_j^2$, where the $r_j$'s are the ellipsoid's principal radii. By Lemma~\ref{lem:ball-containment}, each Macbeath ellipsoid contains a ball of radius $c \SP \eps$, for some constant $c$. Therefore, these diagonal entries are each at least $1/(c \SP \eps)^2$, which implies that its Frobenius norm is $\|M_i\| \leq \sqrt{d}/(c \SP \eps)$. Also, since $\Omega$ has unit diameter, $\|x - c_i\| \leq 1$, and hence $\|M_i (x - c_i) \| \leq \|M_i\| \cdot \| x - c_i \| = O(1/\eps)$. Therefore, $\|\Gradient \psi_i(x)\| = O(1/\eps)$, establishing the first part of the lemma.

To prove the bound on $\|\Gradient \phi_i(x)\|$, recall that $\phi_i(x) = \dfrac{\psi_i(x)}{\Psi(x)}$ and $\Psi(x) = \sum_i \psi_i(x)$. Differentiating, we obtain
\begin{align}
    \Gradient \phi_i(x)
      & ~ = ~ \frac{\Gradient \psi_i(x)}{\Psi(x)} - \frac{\psi_i(x)}{\Psi(x)^2}  \Gradient \Psi(x) 
        ~ = ~ \frac{\Gradient \psi_i(x)}{\Psi(x)} - \frac{\psi_i(x)}{\Psi(x)^2} \cdot \sum_j \Gradient \psi_j(x)  \label{eq:grad_phi} \\
      & ~ = ~ \frac{\Gradient \psi_i(x)}{\Psi(x)} - \frac{\phi_i(x)}{\Psi(x)}\cdot\sum_j \Gradient \psi_j(x). \nonumber
\end{align}
Using the bound on $\|\Gradient \psi_i(x)\|$ from above and the bound on $\Psi(x)$ from Lemma~\ref{lem:positive_Psi}, we proceed to bound $\|\phi_i(x)\|$ as follows
\begin{align*}
    \|\Gradient \phi_i(x)\|
      & ~ \leq ~ \frac{\|\Gradient \psi_i(x)\|}{\Psi(x)} + \frac{\phi_i(x)}{\Psi(x)}\cdot\sum_j \|\Gradient \psi_j(x)\| \\
      & ~ \leq ~ O\left(\frac{1}{\eps}\right) + 1 \cdot O(1) \cdot \depth{\Pi} \cdot O\left(\frac{1}{\eps}\right) 
        ~ \leq ~ O\left(\frac{1}{\eps}\right).
\end{align*}
This establishes the bound on $\|\Gradient \phi_i(x)\|$, completing the proof.
\end{proof}

\subsection{Smooth Distance Gradients} \label{sec:sDF_gradient}

In this section, we establish an explicit expression on the gradient of the smooth distance approximation, which can be easily evaluated in the course of answering distance-function queries, along with an upper bound on its gradient establishing the following lemma.

\begin{restatable}{lemma}{aDmGradient}\label{lem:sDF_gradient}
For the smooth boundary distance approximation encoded in the data structure described in Section~\ref{sec:medial-apx}, we have 
\[
    \Gradient \adm_{\Omega}(x)
        ~ = ~ -\sum_i \phi_i(x) \cdot a_i
          - \frac{4}{\Psi(x)} \sum_i (v_i(x) - \adm_{\Omega}(x)) \cdot\frac{\psi_i(x) f_i(x)}{(f_i(x)^2-1)^2} \cdot M_i(x-c_i).
\]
\end{restatable}

\begin{proof}
Recalling that $\adm(x) = \sum_i \phi_i(x)v_i(x)$, where $\{\phi_i\}$ is the partition of unity and $\{v_i\}$ are the local approximations, we proceed to derive the gradient of $\adm(x)$ as follows.
\[
    \Gradient \adm_{\Omega}(x)
	~ = ~ \Gradient \left(\sum_i \phi_i(x) v_i(x)\right)
        ~ = ~ \sum_i \Gradient\phi_i(x)\cdot v_i(x) + \sum_i \phi_i(x)\cdot\Gradient v_i(x).
\]
For the first term, by Eq.~\eqref{eq:grad_phi} and the fact that $\adm_{\Omega}(x) = \sum_i (\psi_i/\Psi(x)) v_i(x)$ we have
\begin{align*}
    \sum_i \Gradient\phi_i(x) \cdot v_i(x)
	& ~ = ~ \sum_i \left(\frac{\Gradient \psi_i(x)}{\Psi(x)} - \frac{\psi_i(x)}{\Psi(x)^2} \cdot \sum_j \Gradient\psi_j(x) \right) v_i(x) \\
	& ~ = ~ \sum_i \frac{v_i(x)}{\Psi(x)} \Gradient \psi_i(x) - \sum_j \left(\sum_i \frac{\psi_i(x)}{\Psi(x)} v_i(x)\right) \frac{\Gradient\psi_j(x)}{\Psi(x)} \\
	& ~ = ~ \sum_i v_i(x) \frac{\Gradient \psi_i(x)}{\Psi(x)} - \sum_j \adm_{\Omega}(x) \frac{\Gradient\psi_j(x)}{\Psi(x)} \\
	& ~ = ~ \sum_i (v_i(x)-\adm_{\Omega}(x)) \frac{\Gradient\psi_i(x)}{\Psi(x)}. 
\end{align*}
Combining this with Eq.~\eqref{eq:grad_wi_2} yields
\[
    \sum_i \Gradient\phi_i(x) \cdot v_i(x)
	~ = ~ -\frac{4}{\Psi(x)} \sum_i (v_i(x)-\adm_{\Omega}(x)) \cdot \frac{\psi_i(x) f_i(x)}{(f_i(x)^2-1)^2} \cdot M_i(x - c_i),
\]
which matches the first term. For the second term, recall that $v_i(x) = b_i - a_i^{\Transpose} x$, which yields
\[
    \sum_i \phi_i(x) \cdot \Gradient v_i(x)
        ~ = ~ - \sum_i \phi_i(x) \cdot a_i
\]
and completes the proof.
\end{proof}

While we can bound the gradient norm through an analysis of the above expression, there is a simpler analysis using the bounds on the distance approximation.

\begin{restatable}{lemma}{aDmGradBound}\label{lem:sDF_gradient_bound}
For the smooth boundary distance approximation encoded in the data structure described in Section~\ref{sec:medial-apx} for all $x \in \RE^d$, the gradient satisfies $\|\Gradient \adm_{\Omega}(x)\| = O(1)$.
\end{restatable}

\begin{proof}
A key consequence of the partition of unity construction is that the gradients of the normalized weight functions cancel:
\[
    \sum_i \phi_i(x) = 1 
        ~ \implies ~ \sum_i \Gradient \phi_i(x) = 0.
\]
In addition, the absolute error bound required on each local approximation implies that
\[
    v_i(x) 
        ~ = ~ \dm_{\Omega}(x) + \eps_i(x), \qquad \text{where $0 ~\leq~ \eps_i(x) ~\leq~ \eps$, for all $i$.}
\]
Recalling the definition $\adm_{\Omega}(x) = \sum_i \phi_i(x) \cdot v_i(x)$ and by differentiating, we obtain
\[
    \Gradient \adm_{\Omega}(x) ~ = ~ \sum_i v_i(x) \Gradient \phi_i(x) + \sum_i \phi_i(x) \Gradient v_i(x).
\]
We can simplify the first summation by using the cancellation of the weight-function gradients.
\begin{align*}
    \sum_i v_i(x) \Gradient \phi_i(x)
        & ~ = ~ \sum_i \dm_{\Omega}(x) \Gradient \phi_i(x) + \sum_i \eps_i(x) \Gradient \phi_i(x)  \\
        & ~ = ~ \dm_{\Omega}(x) \left(\sum_i \Gradient \phi_i(x)\right) + \sum_i \eps_i(x) \Gradient \phi_i(x)
          ~ = ~ \sum_i \eps_i(x) \Gradient \phi_i(x)
\end{align*}
Recalling that the summation consists of at most $\depth{\Pi}$ non-zero terms and with the aid of Lemma~\ref{lem:blending_gradient_bound}, we can bound the magnitude of the gradient as
\begin{align*}
    \|\Gradient \adm_{\Omega}(x)\| 
        & ~ \leq ~ \sum_i \eps_i(x)\cdot\|\Gradient \phi_i(x)\| + \sum_i \phi_i(x) \cdot \|\Gradient v_i(x)\| \\
        & ~ \leq ~ \depth{\Pi} \cdot \left( \eps \cdot \max_i \|\Gradient \phi_i(x)\| + 1\right) \tag{$\eps_i(x) \leq \eps$ and $\|\Gradient v_i(x)\| = \|a_i\| = 1$}\\
        & ~ \leq ~ \depth{\Pi} \cdot \left( \eps \cdot O\left( \frac{1}{\eps} \right) + 1\right) \tag{Lemma~\ref{lem:blending_gradient_bound}} \\ 
        & ~ = ~ O(\depth{\Pi}).
\end{align*}
By Lemma~\ref{lem:PU_depth}, $\depth{\Pi}$ is $O(1)$, so this is $O(1)$, as desired.
\end{proof}

\subsection{Smooth Distance Hessian}  \label{sec:sDF_hessian_bound}

Next, we establish a bound on the norm of the Hessian of the smooth distance approximation, denoted $\Hess \adm_{\Omega}(x)$.

\begin{restatable}{lemma}{aDmHessBound} \label{lem:sDF_hessian_bound}
For the smooth boundary distance approximation encoded in the data structure described in Section~\ref{sec:medial-apx}, for all $x \in \RE^d$, the Hessian satisfies $\|\Hess \adm_{\Omega}(x)\| = O(1/\eps)$.
\end{restatable}

\begin{proof}
Expressing the Hessian in terms of the directional second derivative, we have
\[
    \| \Hess \adm_{\Omega}(x) \| 
        ~ = ~ \max_{\|v\| = \|u\| = 1} | \Gradient_v \Gradient_u \adm_{\Omega}(x)|,
\]
where $\Gradient_u$ and $\Gradient_v$ denote the directional derivatives in the directions of unit vectors $u$ and $v$, respectively. Using the substitution $\Gradient_u \adm_{\Omega}(x) = \langle u, \Gradient \adm_{\Omega}(x) \rangle$, we rewrite the above as
\begin{align*}
    \| \Hess \adm_{\Omega}(x) \| 
        & ~ = ~ \max_{\|v\| = 1} \max_{\|u\| = 1} | \Gradient_v \langle u, \Gradient \adm_{\Omega}(x) \rangle| \\
        & ~ = ~ \max_{\|v\| = 1} \max_{\|u\|=1} \lim_{\delta \to 0} \frac{|\langle u,  \Gradient\adm_{\Omega}(x + \delta v) - \Gradient\adm_{\Omega}(x)\rangle|}{\delta} \\
        & ~ \leq ~ \max_{\|v\| = 1}\lim_{\delta \to 0} \frac{\|\Gradient\adm_{\Omega}(x + \delta v) - \Gradient\adm_{\Omega}(x) \|}{\delta}.
\end{align*}

To simplify the notation below, we define the following function in order to collect some common terms (which we will further analyze below in Lemma~\ref{lem:F_bounds}).
\begin{equation} \label{eq:F}
    F_i(x) 
        ~ = ~ \frac{4}{\Psi(x)} \cdot (v_i(x) - \adm_{\Omega}(x)) \cdot \frac{\psi_i(x) f_i(x)}{(f_i(x)^2-1)^2}.
\end{equation}

We use the simplified notation to rewrite the explicit gradient from Lemma~\ref{lem:sDF_gradient} as a linear combination of vectors of the following form
\[
    \Gradient \adm_{\Omega}(x) 
        ~ = ~ -\sum_i \phi_i(x) \cdot a_i - \sum_i F_i(x)\cdot M_i(x-c_i).
\]
Applying the definition of the directional derivative, we have
\begin{align*}
    \| \Hess \adm_{\Omega}(x) \| 
        & ~ = ~ \max_{\|v\| = 1} \lim_{\delta \to 0} \frac{\|\Gradient\adm_{\Omega}(x + \delta v) - \Gradient\adm_{\Omega}(x)\|}{\delta} \\
        & ~ \leq ~ \max_{\|v\| = 1} \sum_i \lim_{\delta \to 0} \frac{\|\phi_i(x + \delta v) \cdot a_i - \phi_i(x) \cdot a_i\|}{\delta} \\
        &\qquad\qquad + \sum_i \lim_{\delta \to 0} \frac{\|(F_i(x + \delta v)\cdot M_i (x + \delta v - c_i) - F_i(x) \cdot M_i (x - c_i)\|}{\delta} \\
        & ~ = ~ \max_{\|v\| = 1} \sum_i \lim_{\delta \to 0} \frac{\|(\phi_i(x + \delta v) - \phi_i(x)) \cdot a_i\|}{\delta} \\
        &\qquad + \sum_i \lim_{\delta \to 0} \frac{\|(F_i(x + \delta v) - F_i(x))\cdot M_i (x - c_i) + F_i(x + \delta v) \cdot  M_i (\delta v)\|}{\delta}.
\end{align*}
By straightforward applications of vector norm inequalities and applying the definition of the gradient for $\phi_i$ and $F_i$, we obtain
\begin{align*}
    \| \Hess \adm_{\Omega}(x) \| 
        & ~ \leq ~ \sum_i \| \Gradient \phi_i \| \cdot \|a_i\| + \sum_i \|M_i(x - c_i)\| \cdot \| \Gradient F_i(x) \| \\
        &\qquad + \sum_i \max_{\|v\| = 1}\lim_{\delta\to 0} |F_i(x + \delta v)| \cdot \frac{\|M_i(\delta v)\|}{\delta}.
\end{align*}

Recall that the sum is taken over all overlapping patches of $\Pi$ at $x$, and by Lemma~\ref{lem:PU_depth}, the maximum degree of overlap, denoted $\depth{\Pi}$, is $O(1)$. Also recall that $a_i$ is a unit vector and from Lemma~\ref{lem:blending_gradient_bound} that $\|\Gradient \phi_i(x)\| = O(1/\eps)$. 

In the proof of that lemma, we noted that the terms $\{\|M_i(x - c_i)\|\}$ are upper bounded by $O(1/r_{\min}(M_i))$, where $r_{\min}(M_i)$ is a lower bound on the principal radii of the ellipsoid represented by $M_i$. As shown earlier, $r_{\min}(M_i) \geq c \cdot \eps$ for some constant $c$, and therefore $\|M_i(x - c_i)\| = O(1/\eps)$. It follows from this as well that $\lim_{\delta \to 0} \|M_i(\delta v)\| / \delta$ is bounded above by $1/r^2_{\min}(M_i) \leq O(1/\eps^2)$. As we show below in Lemma~\ref{lem:F_bounds}, $|F_i(x)| = O(\eps)$. Combining these observations, we have
\begin{align*}
    \| \Hess \adm_{\Omega}(x) \| 
        & ~ \leq ~ \sum_i O\bigg( \frac{1}{\eps} \bigg) + \sum_i O\bigg( \frac{1}{\eps} \bigg) \cdot\|\Gradient F_i(x)\| + \sum_i O\bigg( \frac{1}{\eps^2} \bigg) \cdot |F_i(x)| \\
        & ~ =    ~ O\left( \frac{1}{\eps} + \frac{1}{\eps} \cdot \|\Gradient F_i(x)\| + \frac{1}{\eps} \right).
\end{align*}
Per Lemma~\ref{lem:F_bounds} below, $\|\Gradient F_i(x)\| = O(1)$, implying that $\| \Hess \adm_{\Omega}(x) \| = O(1/\eps)$, as desired.
\end{proof}

To finish the analysis, we present the proof of Lemma~\ref{lem:F_bounds}, establishing upper bounds on the magnitudes of both the function value and the gradient of $F_i$.

\begin{lemma} \label{lem:F_bounds}
For any patch $\Pi_i$ and any $x \in \Pi_i$, $|F_i(x)| = O(\eps)$ and $\|\Gradient F_i(x)\| = O(1)$.
\end{lemma}

\begin{proof} Recall the definition from Equation \eqref{eq:F}
\[
    F_i(x) 
        ~ = ~ \frac{4}{\Psi(x)}\cdot\bigg(v(x) - \adm_{\Omega}(x)\bigg)\cdot\frac{\psi_i(x) f_i(x)}{(f_i(x)^2 - 1)^2}.
\]
The bound on the absolute function values follows from the definitions of $v_i$ and $\adm$, both being valid absolute $\eps$-approximations, together with Lemma~\ref{lem:positive_Psi} bounding $\Psi(x)$ from below and Eq.~\eqref{eq:kappa} bounding the last coefficient. Specifically,
\[
    |F_i(x)| 
        ~ \leq ~ \frac{4}{1/4} \cdot \eps \cdot \frac{1}{2} 
        ~ =    ~ O(\eps).
\]
For the gradient bound, we start by taking derivatives.
\begin{align*}
    \Gradient F_i(x)
        & ~ = ~ \bigg(\frac{-4}{\Psi^2(x)}\cdot\sum_j \Gradient \psi_j(x) \bigg)\cdot\bigg(v(x) -     \adm_{\Omega}(x)\bigg)\cdot\frac{\psi_i(x)\cdot f_i(x)}{(f_i(x)^2 - 1)^2} \\
        & ~~ +\frac{4}{\Psi(x)}\cdot\bigg(-a_i - \Gradient\adm_{\Omega}(x)\bigg)\cdot\frac{\psi_i(x)\cdot f_i(x)}{(f_i(x)^2 - 1)^2} \\
        & ~~ + 4 \frac{v(x) - \adm_{\Omega}(x)}{\Psi(x)} \bigg[\frac{\psi_i(x)\cdot (1+3f_i(x)^2)}{(1 - f_i(x)^2)^3} \Gradient \kern-1pt f_i(x) + 4 \frac{\psi_i(x)\cdot f_i(x)^2}{(f_i(x)^2-1)^4} \cdot M_i(x - c_i)\bigg].
\end{align*}
Finally, we substitute the bounds on coefficients and approximation errors. We also apply the bounds from Lemmas~\ref{lem:blending_gradient_bound} and \ref{lem:sDF_gradient_bound}, along with the following simplifying bounds similar to Equation~\eqref{eq:kappa}.
\begin{align*}
    \left| \frac{\psi_i(x)\cdot (1 + 3 f_i(x)^2)}{(1 - f_i(x)^2)^3} \right| 
        & ~ \leq ~ \left| \frac{\mu(f_i(x))\cdot (1 + 3 f_i(x)^2)}{(1 - f_i(x)^2)^3} \right| 
          ~ < ~ 5. \\
    \left| \frac{\psi_i(x)\cdot f_i(x)^2}{(f_i(x)^2-1)^4} \right| 
        & ~ \leq ~ \left| \frac{\mu(f_i(x))\cdot f_i(x)^2}{(f_i(x)^2-1)^4} \right| 
          ~ < ~ 4.
\end{align*}
Together, this yields
\begin{align*}
    \|\Gradient F_i(x) \|
        & ~ \leq ~ O(1)\cdot O(\max_j \|\Gradient \psi_j(x)\|)\cdot \eps\cdot\frac{1}{2} + O(1)\cdot (\|a_i\| + \|\Gradient\adm_{\Omega}(x)\|) \cdot\frac{1}{2} \\
        & \qquad + O(\eps)\cdot\bigg[5\cdot 2\cdot \|M_i(x - c_i)\| + 4 \cdot 4 \cdot \| M_i(x - c_i) \| \bigg] \\
        & ~ \leq ~ \eps\cdot O\left(\frac{1}{\eps}\right) + O(1) + \eps\cdot O\left(\frac{1}{\eps}\right) \\
        & ~ = ~ O(1),
\end{align*}
as desired.
\end{proof}

\section{Concluding Remarks} \label{sec:conclusion}

In this paper, we have taken first steps towards designing data structures for approximately answering geometric distance queries approximately, while more faithfully preserving properties of the underlying distance functions. Existing data structures based on computing approximate nearest neighbors suffer from discontinuities in the resulting distance field, which is undesirable in many applications. We have presented a general method for achieving smoothness by combining a traditional (discontinuous) method with blending,
and we have illustrated the technique in the concrete application of approximating (in terms of absolute errors) the distance field to the boundary, induced within a convex polytope $\Omega$ in $\RE^d$. Our data structure is efficient in the sense that it nearly matches the best asymptotic space and time bounds for the simpler problem of approximately determining membership within the polytope (being suboptimal by a factor of $1/\sqrt{\eps}$ in the space). We have also presented bounds on the norms of the gradient (first derivative) and Hessian (second derivative) of the approximation.

There are a number of interesting open problems that remain. The first is applying this method to more approximate nearest neighbor search applications. We have done this for a discrete set of points in $\RE^d$, which we plan to publish in a future paper. While our results nearly matching the best known complexity bounds for $\eps$-approximate nearest neighbor searching, the technical issues are quite involved. The method can be applied to other query problems where the answer is naturally associated with a continuous field. Examples include penetration depth in collision detection~\cite{ZHANG20143}, distance oracles in robotics and autonomous navigation \cite{WWL22}, and novel-view synthesis using parametric radiance fields~\cite{Fridovich-Keil_2022_CVPR}.

While our approach produces a smooth approximation, there are other properties of distance fields that would be useful to preserve. One shortcoming of our method is that it can produce spurious local minima in the approximate distance field. An interesting question is whether our approach can be modified to eliminate these minima. We anticipate interesting connections to the literature on vector field design~\cite{vaxman2016directional} and mode finding~\cite{LLM21,EFR12}.



\end{document}